\documentclass[a4paper,11pt,accepted=2021-01-14]{quantumarticle}
\pdfoutput=1

\usepackage[T1]{fontenc}
\usepackage[utf8]{inputenc}
\usepackage{amsmath}
\usepackage{amsthm}
\usepackage{amsfonts}
\usepackage{bbm}
\usepackage{bm}
\usepackage{bbold}
\usepackage{amssymb}
\usepackage{graphicx,epsfig,color,tcolorbox}
\usepackage{physics}
\usepackage[sort&compress, numbers]{natbib}
\usepackage{mathtools}
\usepackage{tabularx}
\usepackage[colorlinks=true,linkcolor=teal,citecolor=purple]{hyperref}

\newcommand{\plb}{\textcolor{black}}
\newcommand{\ot}{\otimes}
\newcommand{\diag}{\mathrm{diag}}
\newcommand\numberthis{\addtocounter{equation}{1}\tag{\theequation}}

\newtheorem{theorem}{Theorem}

\newtheorem{construction}{Construction}
\newtheorem{lemma}{Lemma}
\newtheorem{corollary}{Corollary}
\newtheorem{definition}{Definition}

\theoremstyle{definition}
\newtheorem{example}{Example}

\begin{document}
\title{Second law of thermodynamics for batteries with vacuum state}

\author{Patryk Lipka-Bartosik}
\email{patryk.lipka.bartosik@gmail.com}
\affiliation{
 H. H. Wills Physics Laboratory, University of Bristol, Tyndall Avenue, Bristol, BS8 1TL, United Kingdom
}
\affiliation{Institute of Theoretical Physics and Astrophysics,
National Quantum Information Centre, Faculty of Mathematics, Physics and Informatics,
University of Gdańsk, Wita Stwosza 57, 80-308 Gdańsk, Poland}
\author{Paweł Mazurek}
\affiliation{Institute of Theoretical Physics and Astrophysics,
National Quantum Information Centre, Faculty of Mathematics, Physics and Informatics,
University of Gdańsk, Wita Stwosza 57, 80-308 Gdańsk, Poland}
\affiliation{International  Centre  for  Theory  of  Quantum  Technologies,University  of  Gdansk,  Wita  Stwosza  63,  80-308  Gdansk,  Poland}
\author{Michał Horodecki}
\affiliation{International  Centre  for  Theory  of  Quantum  Technologies,University  of  Gdansk,  Wita  Stwosza  63,  80-308  Gdansk,  Poland}

\date{1 February 2021}

\begin{abstract}
In stochastic thermodynamics work is a random variable whose average is bounded by the change in the free energy of the system. In most treatments, however, the work reservoir that absorbs this change is either tacitly assumed or modelled using unphysical systems with unbounded Hamiltonians (i.e. the ideal weight). In this work we describe the consequences of introducing the ground state of the battery and hence --- of breaking its translational symmetry. The most striking consequence of this shift is the fact that the Jarzynski identity is replaced by a family of inequalities. Using these inequalities we obtain corrections to the second law of thermodynamics which vanish exponentially with the distance of the initial state of the battery to the bottom of its spectrum. Finally, we study an exemplary thermal operation which realizes the approximate Landauer erasure and demonstrate the consequences which arise when the ground state of the battery is explicitly introduced. In particular, we show that occupation of the vacuum state of any physical battery sets a lower bound on fluctuations of work, while batteries without vacuum state allow for fluctuation-free erasure.
\end{abstract}

\maketitle
\section{Introduction} 
\label{intro}
The second law of thermodynamics sets limits for all physical processes. It determines which state transformations are possible, regardless of the microscopic details of the process. From a practical point of view it imposes fundamental restrictions on the amount of average work $\langle w \rangle $ performed by the system which evolves from the state $\rho$ towards $\rho'$ and interacts with a thermal reservoir at a fixed temperature, i.e.:
\begin{align}
 \label{eq:mt1}
\langle w \rangle \leq F(\rho) - F(\rho') = -\Delta F,
\end{align}
where $F(\rho)$ is the system's free energy which will be defined later. The second law of thermodynamics is a statistical law and as such, it governs how thermodynamic systems behave when averaged over many realizations of the experiment. This information is relevant for a macroscopic observer, however, it does not provide much information about the microscopic details of the occuring process.

Recent developments in experimental techniques allow for manipulating and measuring systems at the nanoscale level \cite{Koski2014,Koski2015,Chida2017,Camati2016,Peterson2016,Zanin2019}. In order to take full advantage of these techniques it is crucial to understand how thermodynamic laws translate into the non-equilibrium domain, where fluctuations of thermodynamic quantities begin to play a significant role and averaged quantities are no longer enough to characterize their thermodynamic behaviour. This motivates extending thermodynamic framework to systems driven out of equilibrium, a setting which has been extensively studied in the recent literature \cite{Jarzynski1997,Jarzynski1997a,Alhambra2016,Richens2016,Ito2018,Debarba2019,Rao2018,lobejko2020tight,purves2020channels}.

Fortunately, there exist much stronger constraints on the possible distributions of thermodynamic work than the second law of thermodynamics (\ref{eq:mt1}). These constraints are often referred to as ``fluctuation relations''. Arguably the best-known relation from this family is the \emph{Jarzynski equality} which implies that the probability of extracting work per particle larger than the free energy difference (\ref{eq:mt1}) is exponentially suppressed with the number of particles in the system \cite{Jarzynski1997}.The Jarzynski equality has been thoroughly studied both theoretically \cite{Cohen2004,Palmieri2007} and experimentally \cite{Douarche2005,An2015,Hoang2018} and extended to the quantum case \cite{kurchan2000quantum,tasaki2000jarzynski,Alhambra2016}.

A standard and often implicitly accepted assumption required for the Jarzynski equality to hold is the existence of a perfect work receiver. This allows to  treat the whole free energy difference as thermodynamic work. However, in some cases the existence of such a perfect receiver is difficult to motivate --- especially in the quantum regime, where the device  has to be modelled itself as a quantum system. The common approach is to consider the explicit work receiver (or the ``battery'') as an ancillary system which interacts with the working body and the heat bath during the transformation. In such situations the battery device is often modelled using either a classical Hamiltonian (e.g. \cite{PhysRevX.3.041003}) or a Hamiltonian unbounded from below (e.g. \cite{Alhambra2016}). Both of these approaches, although perfectly valid in the classical regime, cannot be justified when battery's energy is close to its ground state. Therefore, in order to understand the effects arising in small-scale thermodynamics, it is important to understand how the presence of the battery ground state influences thermodynamic protocols. 

This approach is fundamental to the modern (resource-theoretic) program of thermodynamics which attempts to systematically account for all possible resources involved in the process. One then studies their conversions under a restricted set of physically motivated interactions between the system, the heat bath and the battery device. In the most general case this interaction is modelled using an energy-preserving unitary acting on these three subsystems. Thermodynamic work is then defined by specifying an ``explicit battery model'', which amounts to ($i$) specifying the Hamiltonian of the ancillary battery system, ($ii$) the work quantifier (e.g. average energy) and ($iii$) the allowed interactions between the battery, the system and the bath (or restrictions on the allowable battery states). Two most widely used models in the literature involve a qubit battery (wit) where the battery is in a pure state at all times and work is defined using average energy, and an ideal weight battery. In this model work is also defined using average energy of the battery but the global unitary is now assumed to commute with the translations on the battery, leading to the notion of \emph{translational invariance} (TI). This is, however, only possible when the energy spectrum of the battery ranges over the whole real line, i.e. the battery is unbounded from both above and below. The ideal weight model is a common way of defining work in the quantum regime which was considered for the first time in \cite{Skrzypczyk2014} and then utilized to prove several fundamental results in the field of quantum thermodynamics \cite{Alhambra2016,Aberg2014,Aberg2018,Masanes2017,Richens2016}. Importantly, this additional assumption has powerful physical implications: it assures that work satisfies the second law of thermodynamics \cite{Masanes2017} and furthermore leads to the Jarzynski equality \cite{Alhambra2016}.

Naturally, one can argue that the ideal weight is not a realistic model of a battery as it does not have a ground state energy. On the other hand, one can also argue that if the transformation is performed sufficiently far from the vacuum, the evolution assisted with a battery with a ground state (which we will refer to as the \emph{physical battery}) should be equivalent to the evolution assisted with the ideal weight \cite{Aberg2014,Faist2016} and should reproduce the same quantitative results (i.e. work distributions) as a physical battery. However, the Nature we observe often does not follow this scheme; the existence of the ground state can be very often perceived, no matter how far we are from it. A basic question then appears: does the existence of the vacuum state of the battery have any implications for thermodynamic processes? 

Here we study this question in detail. We examine two different regimes of battery operation: the \emph{high-energy} regime in which the average energy of the battery is far from the ground state energy (i.e. population of the ground state is small) and \emph{low-energy} regime in which the occupation of the ground state cannot be ignored. While we corroborate the intuition that the battery with vacuum essentially behaves like the ideal weight in the high energy regime, we investigate quantitatively different predictions to which it leads in the low-energy regime, in particular with respect to Jarzynski equality, second law of thermodynamics and fluctuations of work.

In the first part of the paper we introduce a generalization of the translational invariance property which we term \emph{effective translational invariance} (ETI). This allows us to consider thermodynamic processes with broken translational symmetry and hence provides a convenient tool for studying quantum thermodynamics using arbitrary physical batteries.

We further show that ETI in the framework of thermal operations (TO) implies a generalization of the standard Jarzynski equality, i.e. a family of inequalities which impose looser constraints on the allowable work distributions than the standard Jarzynski equality. This is our first main result.

Using these new inequalities we show that one can still recover the second law of thermodynamics (\ref{eq:mt1}) in an approximate form, i.e. with correction terms decaying exponentially fast with the distance to the ground state. This allows to identify heat contributions to the average energy change when the battery is bounded from below and shows that global translational symmetry is not necessary to satisfy the second law of thermodynamics. Importantly, these deviations are related to the average work rather than just work and hence they are of a completely different character than the deviations reported by standard fluctuation theorems. This is the second main contribution of our paper. 

In the next part of the paper we show that our model correctly reproduces the single-shot results on the work of formation originally derived using qubit as the battery system \cite{Horodecki2013}. This answers an open problem from the field of quantum thermodynamics by showing that the notions of \emph{single-shot} and \emph{fluctuating} work can be both properly defined and studied for a battery with a ground state. This is the third main result of the paper.

In the last part of the paper we study a paradigmatic example of Landauer erasure and compare how the two battery models (the ideal weight battery and the harmonic oscillator with vacuum) behave when assisting this type of transformation. Notice that the second case contrasts with the classical treatment of work as the battery can no longer supply arbitrary values of energy when it operates close to its vacuum state. Because of this both batteries experience quantitative differences in their behavior in the low-energy regime. We show that for any battery model the very existence of the vacuum state implies a lower bound on the size of work fluctuations whenever energy is taken from the battery. In this way the occupation of the vacuum state is the fundamental factor which forbids performing Landauer erasure with deterministic (or arbitrarily concentrated) work distribution. We further show that the ideal weight violates this bound and hence allows for transformations which cannot be achieved in the low-energy regime.

We finish the paper with a short summary and present several related open problems which we believe to be relevant for the field of quantum thermodynamics.

\section{Framework}
\label{framework}
Throughout we adapt a resource-theoretic approach to quantum thermodynamics called \emph{thermal operations} \cite{Ruch1976,Horodecki2013,Richens2017,Brandao2013,Lostaglio2016,Perry2018,Faist2015minimal,Masanes2017,Goold2015,Masanes2017,Faist2015gibbs,Lostaglio2015,Korzekwa2016,Horodecki2014,Kwon2018clock,Mueller2018,obejko2020thermodynamics,aguilar2020thermal,boes2020variance}. This is a well-established framework for studying thermodynamic processes in the quantum regime which gives the experimenter the most freedom in manipulating systems without access to external resources like coherence or asymmetry \cite{Aberg2014,Lostaglio2018,Marvian2018no,Marvian2016quantum,Gour2018quantum}, entanglement \cite{Chitambar2016,Contreras2018,Bera2017laws} or conserved quantities \cite{Lostaglio2017,Mingo2018,Guryanova2016}. It is thus a convenient class of operations for deriving fundamental thermodynamic limitations. Moreover, these operations are also experimentally achievable as they can be realized with a very coarse-grained control \cite{Perry2018}. A very readable and comprehensive introductions to this field of quantum thermodynamics can be found in \cite{Lostaglio2019} and \cite{binder2019thermodynamics}. See also \cite{Guarnieri_2019} for an attempt to reconcile this framework with a more standard approach of stochastic thermodynamics.   

The setting consists of a system $\rm S$ with Hamiltonian $H_{\rm S}$ and an infinite heat bath $\rm{B}$ with Hamiltonian $H_{\rm{B}}$ satisfying a few reasonable assumptions (see \cite{Horodecki2013} for the details) initially in a Gibbs state $\tau_{\rm{B}} = e^{-\beta H_{\rm{B}}}/Z_{\rm{B}}$, where $Z_{\rm{B}} = \tr e^{-\beta H_{\rm{B}}}$ is the partition function. The interaction of the system $\rm{S}$ with the heat bath $\rm{B}$ is modelled using a unitary transformation $U$ and hence the effective map on $\rm S$ can be represented as a channel:
\begin{align}
\label{eq:mt2}
\Gamma_{\rm{S}} \left[ \rho_{\rm S} \right] = \tr_{\rm B} \left[ U \left( \rho_{\rm S} \otimes \tau_{\rm B} \right) U^{\dagger} \right],
\end{align}
with $\rho_{\rm S}$ being the initial state of the system. To conserve the total energy of the system for every possible initial state, the unitary $U$ must commute with the total Hamiltonian, i.e. $[U, H_{\rm S} + H_{\rm B}] = 0$. This can be also thought of as a microscopic statement of the first law of thermodynamics. We will refer to the effective map (\ref{eq:mt2}) as a \emph{thermal operation} (TO). Thermal operations are general enough to account for the situations in which the Hamiltonian of the system changes from $H_{\rm S}$ to $H_{\rm S'}$. This can be done by adding an ancillary qubit system which acts as a switch as was described in \cite{Horodecki2013}. For further details see also the Appendix.

One of the most fundamental problems studied by quantum thermodynamics is determining conditions on states $\rho_{\rm S}$ and $\sigma_{\rm S}$ such that there exists a feasible unitary in (\ref{eq:mt2}) which realizes the transformation, which we write as $\rho_{\rm S}\small{ \xrightarrow{\rm TO}} \sigma_{\rm S}$. In the most general case of arbitrary quantum states the necessary and sufficient set of conditions for the existence of a TO realizing $\rho_{\rm S} \small{ \xrightarrow{\rm TO}}  \sigma_{\rm S}$ has not been found yet, though there are a few important partial results \cite{Cwiklinski2015,Lostaglio2015}. The situation greatly simplifies when states $\rho_{\rm S}$ and $\sigma_{\rm S}$ are incoherent in the energy eigenbasis. In particular, when the states are incoherent and the Hamiltonian is fully degenerate, i.e. $H_{\rm S} = 0$, the convertibility under thermal operations is fully governed by majorization \cite{hardy1952}. Namely, for two states $\rho_{\rm S}$ and $\sigma_{\rm S}$, $\rho_{\rm S}\small{ \xrightarrow{\rm TO}}  \sigma_{\rm S}$ if and only if the vector $\bm{q} = \diag[\sigma_{\rm S}]$ is majorized by the vector $\bm{p} = \diag[\rho_{\rm S}]$. This is equivalent to the existence of a doubly-stochastic matrix $R$ which transforms $\bm{p}$ into $\bm{q}$, i.e. $R\, \bm{p} = \bm{q}$. Once the doubly-stochastic matrix $R$ is found, the unitary from (\ref{eq:mt2}) can be easily deduced as described e.g. in \cite{Korzekwa2016coherence}.

When the energy spectrum of the Hamiltonian is non-trivial, TOs preserve the associated Gibbs state rather than the uniform state. In this case majorization is replaced by thermo-majorization \cite{Janzing2000,Horodecki2013} which is equivalent to the existence of a stochastic matrix which preserves the vector of Gibbs weights $\mathbf{g} = \diag[\tau_{\rm S}]$ and transforms $\mathbf{p}$ into $\mathbf{q}$. This leads to a whole new set of conditions which determine which state transformations are possible under the laws of thermodynamics and can be thought of as a single-shot refinement of the second law of thermodynamics. If coherences between different energy eigenstates of the system are present, these conditions are necessary, but not sufficient.

It is natural to ask how these conditions modify when we allow for performing thermodynamic work. Arguably the most studied modification of the framework is to add an ancillary battery system $\rm W$ and study transformation of the form (\ref{eq:mt2}) where system $\rm S$ is now replaced by the joint system $\rm SW$. One then identifies thermodynamic work with the energy change on this ancillary system. More precisely, before and after applying the global unitary $U$ the energy of the battery system is measured using projective measurements, obtaining outcomes $\epsilon$ and $\epsilon +w$ respectively. Work $w$ is then a random variable with probability distribution given by:
\begin{align}
\label{eq:mt_work_def}
p(w) = \int_{\epsilon} \text{d} \epsilon \, \tr \Big[&\left( \mathbb{1}_{\rm S} \otimes \Pi_{\epsilon+w} \otimes \mathbb{1}_{\rm B}  \right) \, \!\times\! \nonumber \\ &U ( \rho_{\rm S} \otimes  \Pi_{\epsilon} \, \rho_{\rm W}\, \Pi_{\epsilon} \otimes \tau_{\rm B} )  U^{\dagger} \Big]
\end{align} 
where $\Pi_{x} = \dyad{x}_{\rm W}$ is a projector onto one of the energy eigenstates of $\rm W$.

This way of defining work leads to two fundamental questions: ($i$) when can we treat work stored in the battery as legitimate thermodynamic work and ($ii$) what are the necessary and sufficient conditions for the \emph{work-assisted} state transformations $\rho_{\rm S} \small{ \xrightarrow{\rm TO}} \sigma_{\rm S}$? It is perhaps surprising that answering those questions relies on which model of the battery we use. In the next sections we will describe and discuss properties of the main battery models found in the literature (see also \cite{PhysRevB.99.035421,PhysRevLett.122.047702,PhysRevLett.120.117702,binder2015quantacell,julia2020bounds}).

\subsection{Thermal operations with qubit battery}
The simplest approach proposed in \cite{Horodecki2013} is to use as a battery a two-level qubit (wit) with a tunable energy gap $\delta$ and a Hamiltonian:
\begin{align}
H_{\rm W} = \delta \dyad{1}_{\rm W}    
\end{align}
Furthermore, one then assumes that the battery always starts either in the ground or the excited state. This simple construction allows to introduce a notion of work valid in the single-shot regime. The arising work is often referred to as ``deterministic work'' and is defined as the optimal $\delta$ for which there exists a TO realizing:
\begin{align}
\label{eq:TOwit}
\rho_{\rm S} \otimes \dyad{i}_{\rm W} \xrightarrow{\rm TO} \sigma_{\rm S} \otimes \dyad{j}_{\rm W}.
\end{align}
where $(i, j) = (0, 1)$ when the work is stored in the battery (distillation) and $(i, j) = (1, 0)$ when it is consumed (formation). Since the battery starts and ends up in a pure energy eigenstate, its energy difference can be fully associated with the work that needs to be performed (or can be extracted) during the transformation. To determine whether a given state $\rho_{\rm S}$ can be converted into $\sigma_{\rm S}$ using $\delta$ of work one has to check the thermo-majorization criteria between the joint states of $\rm SW$ (see Appendix A for the details).

Naturally, transformations in which the wit ends up in a pure state are practically impossible to achieve. That is why for a theory to be applicable to realistic protocols one should consider transformations of the form (\ref{eq:TOwit}), where now the battery is allowed to finish in a slightly mixed state. However, at this point it is not clear if we can interpret the average energy change of the wit as valid thermodynamic work. Since its entropy has changed, it is impossible to differentiate it between work and heat. To see this more explicitly, consider the following example.

\begin{example}{(\emph{Thermalization of a wit})}
\label{ex1}
Consider a process which outputs a joint Gibbs state irrespective of the input, i.e.:
\begin{align}
    \forall\, \rho_{\rm S}, \rho_{\rm W} \qquad \rho_{\rm S} \otimes \rho_{\rm W} \xrightarrow{\rm TO} \tau_{\rm S} \otimes \tau_{\rm W}.
\end{align} 
Notice that this is a valid thermal operation as it always preserves the total Gibbs state. Consider applying this map to the state $\rho_{\rm S} \ot \rho_{\rm W} = \tau_{\rm S} \ot \dyad{0}_{\rm W}$. The average work associated with this transformation is given by:
\begin{align}
 \nonumber
\langle w \rangle &= \tr \left[ H_{\rm W} \left(\tau_{\rm W} -  \dyad{0}_{\rm W} \right)\right] \\  &= \frac{\delta}{1+e^{\beta \delta}} \geq 0.
\end{align}
Hence, thermalization of a wit yields a positive amount of work on average, in contradiction to the second law of thermodynamics (\ref{eq:mt1}). Notice that the above example does not show that wit is a deficent battery model, but rather that it does not behave well in a non-ideal scenario. In particular, allowing for changing its entropy (however slightly) poses certain difficulties in interpreting its average energy change as thermodynamic work.
\end{example}

\subsection{Thermal operations with ideal weight battery}
To resolve the problem we observed in the previous section we can use a different battery model, i.e. an ideal weight with Hamiltonian:
\begin{align}
H_{\rm W} = \int_{-\infty}^{\infty}  x \dyad{x}_{\rm W} \text{d} x  
\end{align}
where the basis $\left\{ \ket{x}_{\rm W} |\, x\in \mathbb{R} \right\}$ is formed from continuous orthonormal states representing the position of the weight. Notice that assuming this particular Hamiltonian is in some sense an arbitrary choice. We can also choose to store work in the kinetic energy of a moving particle or angular momentum of a rotating wheel. What is important is that the model is doubly-infinite, i.e. it has the capacity to store and provide arbitrary amounts of work for any possible transformation. This is of course not a realistic assumption as every physical implementation of the ideal weight must have an energetic minimum, as well as moving particles and rotating wheels eventually halt. 

The unitary $U$ from (\ref{eq:mt2}) is assumed to commute with the shift operator on the weight defined as $\Delta_y := \int_{-\infty}^{\infty} \dyad{x+y}{x}_{\rm W} \text{d} x$, i.e.:
\begin{align}
    \label{eq:com_rel}
    \forall\,{y} \qquad [U_{\rm SBW}, \, \,\text{id}_{\rm SB} \ot \Delta_y^{}] = 0,
\end{align}
This additional assumption is often referred to as \emph{translational invariance} (TI) and its importance was highlighted in \cite{Masanes2017} where it was shown that this combined with the energy conservation assures that dumping entropy into the joint state of the system and the battery (as we saw in Example \ref{ex1}) is impossible. TI is also enough to recover not only the second law of thermodynamics, but also quantum versions of the Jarzynski and Crooks fluctuation theorems \cite{Alhambra2016}. These results provide solid grounds to interpret shifts on the ideal weight as legitimate thermodynamic work.  Moreover, the authors of \cite{Alhambra2016} also proved that the necessary and sufficient conditions for a work-assisted transformation between incoherent states in this case can be expressed by a generalization of thermo-majorization which they termed \emph{Gibss-stochasticity}, i.e.:
\begin{align}
\label{mt:7}
\forall\, {s'}\qquad \sum_{s,w} p(s',w|s) e^{\beta (\widetilde{E}_{s'}-E_s+w)} = 1,
\end{align}
where $p(s', w|s)$ is the conditional probability of the final state of the system having energy levels $\widetilde{E}_{s'}$ (energy of the final Hamiltonian $H_{\rm S'}$) and work $w$ being done by the system, given that the initial state of the system had energy level $E_s$ (energies of the initial Hamiltonian $ H_{\rm S}$). Interestingly, by setting $w = 0$ one can recover from (\ref{mt:7}) the thermo-majorization criteria, as described in \cite{Alhambra2016}. 

However, even though this approach is convenient mathematically and recovers the standard thermodynamic results, this proposal leads to new problems. First, Nature does not allow for Hamiltonians which are unbounded from below and thus it is not clear how well the ideal weight model describes any physical battery. Secondly, the conditions (\ref{mt:7}) lead to certain limitations when one tries to study deterministic work, as described by the following example: 

\begin{example}{(\emph{No deterministic work})}
\label{dwexample}
Consider the following process realized using the ideal weight and a qubit system $\rm S$ with a fully-degenerate Hamiltonian $H_{\rm S} = 0$, i.e.:
\begin{align}
    \label{eq:ex2_trans}
    \forall \rho_{\rm S} \qquad \rho_{\rm S} \xrightarrow{\rm TO} a \dyad{0}_{\rm S} + b \dyad{1}_{\rm S},
\end{align}
where $a$, $b$ are positive such that $a+b = 1$. When explicitly including the ideal weight we have that for all $\rho_{\rm S}$ the input state $\rho_{\rm S} \otimes \dyad{0}_{\rm W}$ is transformed into:
\begin{align}
\sum_{s,s',w} p_{\rm S}(s)\, p(s', w|s) \dyad{s'}_{\rm S} \otimes \dyad{w}_{\rm W},
\end{align}
where $\rho_{\rm S} = \sum_s p_{\rm S}(s) \dyad{s}_{\rm S}$ and we used the fact that due to TI we can start in an arbitrary energy eigenstate of the battery. Now the probability distribution $p(s', w|s)$ reads: 
\begin{align}
    p(s', w|s) = p(s'|s)\, p(w|s,s').
\end{align}
In our particular example we have $p(s' = 0|s) = a$ and $p(s' = 1|s) = b$ for all $s \in \{0, 1\}$. However, notice that if we now demand a deterministic work cost, i.e. $p(w) = \delta_{w, w^*}$ for some real number $w^*$, then conditions (\ref{mt:7}) imply: 
\begin{align}
    s' = 0:& \qquad w_{^*} = - k_{\text{B}} T \log{2} - k_{\text{B}} T \log a, \\
    s' = 1:& \qquad w_{^*} = - k_{\text{B}} T \log{2} - k_{\text{B}} T \log b.
\end{align}
This can only be satisfied if $a = b = 1/2$ and hence shows that it is impossible to obtain a fluctuation-free work when performing (\ref{eq:ex2_trans}) using an ideal weight battery.
\end{example}

It is tempting to think that this behavior is a consequence of the third law of thermodynamics. However, note that in this case we have access to two infinitely big systems: an infinite heat bath and an (infinite) ideal weight and we put no constraints on how many degrees of freedom we may access. Notice also that if we lift the TI constraint and choose two distinct energy levels of the battery separated by an energy difference $\delta$ then we are able to exactly recover any deterministic transformation of the form (\ref{eq:TOwit}) (see e.g. \cite{Faist2016} for a proof of this statement). Hence the problem must be somehow related to the TI property. Surprisingly, we will see that by assuming a less demanding notion of translational invariance we can recover both a proper behavior in the macroscopic limit (second law) and the desired behavior in the microscopic limit (deterministic work).

\subsection{Thermal operations with harmonic oscillator battery}
Motivated by these realizations we now consider a battery model which has an energy spectrum bounded from below. Arguably the simplest model which satisfies this property is a harmonic oscillator:
\begin{align}
    H_{\rm W} = \sum_{k=0}^{N} \epsilon_k \dyad{\epsilon_k}_{\rm W},
\end{align}
with $\epsilon_k := k\, \delta $ and $N$ is the number of energy levels of the oscillator. For diagonal states any thermal operation $\Gamma_{\rm SW}$ acting on a harmonic oscillator battery can be fully characterized by a set of transition probabilities $\{r(s'k'|sk)\}$ which describe the probability that the state $\ket{s}_{\rm S} \otimes \ket{\epsilon_k}_{\rm W}$ gets mapped to another state $\ket{s'}_{\rm} \otimes \ket{\epsilon_{k'}}_{\rm W}$. They can be extracted from $\Gamma_{\rm SW}$ via: 
\begin{align}
    \label{mt:rcomp}
    r(s'k'|sk) := \tr\Big(\dyad{s'} &\otimes \dyad{\epsilon_{k'}} \times \\ &\Gamma_{\rm SW} \big[ \dyad{s} \otimes \dyad{\epsilon_k} \big]\Big). \nonumber 
\end{align} 
For $\Gamma_{\rm SW}$ to be a valid thermal operation the associated transition probabilities must necessarily satisfy:
\begin{align}
\label{eq:mt_stoch}
\forall s, k \hspace{5pt} &\sum_{s', k'} r(s'k'|sk) = 1, \\
\label{eq:mt_gpcond}
\forall s', k' \hspace{5pt}  &\sum_{s, k} r(s'k'|sk)e^{\beta(\widetilde{E}_{s'} - E_s + \epsilon_{k'} - \epsilon_k)} = 1.
\end{align}
The first line of these conditions means that $\Gamma_{\rm SW}$ is a trace-preserving map whereas the second line assures that the channel preserves the Gibbs state. Note that for a battery with the bottom the commutation relation (\ref{eq:com_rel}) cannot be satisfied for all $y$ (unless the map is the identity map), and therefore a more general notion of translational invariance is needed. In many physically relevant situations it is impossible to have precise control over the the unitary from (\ref{eq:mt2}) and hence it may be very difficult to impose the TI condition in practice. Here, instead of constraining the global unitary, we will put constraints on the effective map on the system and the battery arising from this unitary. This is a much looser constraint than (\ref{eq:com_rel}) and is potentially easier to implement in practice. 

Let us consider a discrete energy translation operator acting on the harmonic oscillator battery $\Delta_n = \sum_{k \geq n} \dyad{k-n}{k}_{\rm W}$. We are going to assume that the thermal operation $\Gamma_{\rm SW}$ is invariant with respect to translations of the battery only above a certain threshold energy $\epsilon_{\text{min}} := \delta \, k_{\text{min}}$ for some $0 \leq k_{\text{min}} \leq N$. We will refer to this notion as \emph{effective translational invariance} (ETI). Formally this means that the channel commutation relation:
\begin{align}
     \label{mt:eti}
     \qquad \Big[\Gamma_{\rm SW}, \text{id}_{\rm S} \ot \Delta_n [\cdot]\Delta_n^{\dagger}\Big] \left[\rho_{\text{SW}}\right] = 0,
\end{align}
holds for all states of the battery above the threshold energy $\epsilon_{\text{min}}$, i.e. $\forall \, \rho_{\rm SW}$ s.t. $\tr[(\mathbb{1}_{\rm S} \ot \Pi_{\epsilon}) \rho_{\rm SW}] = 0$ for $\epsilon < \epsilon_{\text{min}}$. Notice that this condition is less stringent than the commutation relation (\ref{eq:com_rel}), i.e. TI implies ETI but the converse is not true in general. In terms of transition probabilities $\{r(s'k'|sk)\}$ this condition can be restated as (see Appendix \ref{appB}): 
\begin{align}
\label{eq:mt5}
r(s' k'|s k) = r(s', k'+n|s, k+n),
\end{align} 
for all integer $n$ such that $0 \leq k'+n \leq N$ and $k_{\text{min}} \leq k+n \leq N$. We will refer to the set of energy levels below and above $\epsilon_{\text{min}}$ as the \emph{vacuum} and \emph{invariant regime} respectively. Intuitively, the ETI assumption means that whenever a given process has a non-zero probability of taking the battery from $\ket{\epsilon_k}_{\rm W}$ to $\ket{\epsilon_{k'}}_{\rm W}$, then all transitions related to the associated work cost $w = \epsilon_{k'} - \epsilon_k$ are equally probable. For a graphical explanation see Fig. \ref{invariance}.

\begin{figure}[h!]
\includegraphics[width=\linewidth]{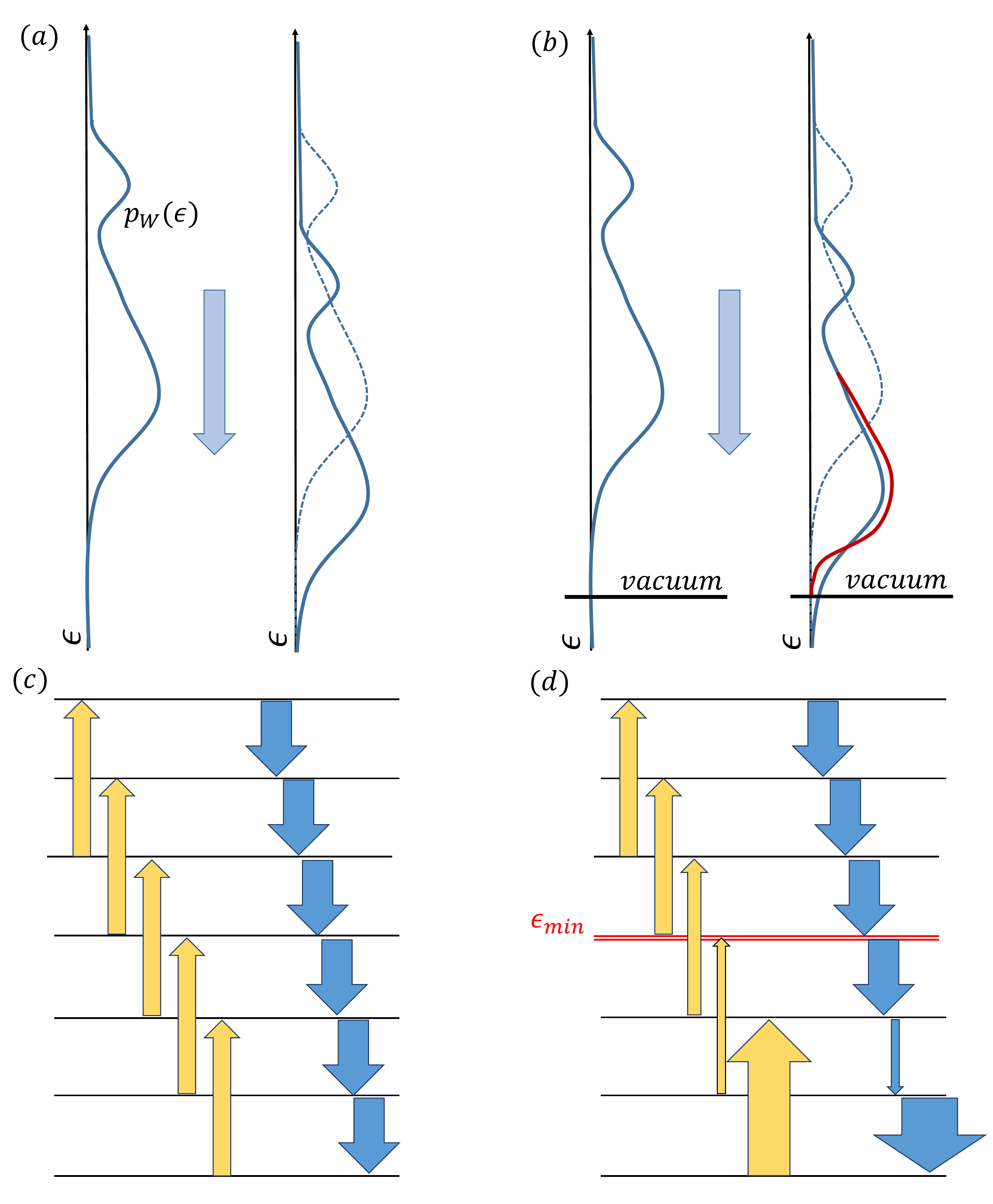}
\caption{A battery with unbounded spectrum (a), attached to the system and environment, works as a tool to define work by changes of its average energy $\langle w\rangle$ during operations described by unitaries applied to the system, battery and environment. For a given transition between system states, $\ket{s}_{\rm S} \rightarrow \ket{s'}_{\rm S}$, transitions between selected energy levels of the battery are represented by arrows (c), (d) -- the width of an arrow is proportional to the probability value. A common assumption that the unitaries commute with the shift operator on the unbounded battery (a) leads to the conclusion that probabilities of transitions on the battery are the same for the same energy gain on the battery, disregarding what is its initial state (c). For a physical model of a battery bounded from below, unitaries cannot commute with the shift operator, as the presence of the vacuum affects the final distribution of battery populations (b). We model this by allowing transitions emerging from levels below $\epsilon_{\text{min}}$ to break translational invariance (d). This introduces corrections to the standard second law inequality for $\langle w\rangle $, which shows that the average change of energy of the battery ceases to serve as a good measure of work. While the corrections vanish exponentially with the distance between $\epsilon_{\text{min}}$ and the bulk of $p_{\rm W}(\epsilon)$, the presence of the vacuum can be exploited to reduce work fluctuations, allowing for (almost) perfect Landauer erasure on physical batteries.} 
\label{invariance}
\end{figure}

Let us note that if the initial state of the system and the battery is not diagonal in the energy basis, then $\Gamma_{SW}$ is not fully characterized by  (\ref{mt:rcomp}). However, (\ref{eq:mt_stoch}) and (\ref{eq:mt_gpcond}) still have to be satisfied, irrespectively of the evolution which $\Gamma_{SW}$ sets for off-diagonal elements of the density matrix. As work is determined only by changes of occupations of the battery energy levels, in the following sections we 
exploit  (\ref{eq:mt_stoch}) and (\ref{eq:mt_gpcond}) to derive a fluctuation relation for work and a formulation of the second law of thermodynamics valid in the ETI model (\ref{eq:mt5}). The only alternation that needs to be taken into account is the fact that functions $f_s$ and $F$ which we use in the proofs are related to the free energy only for the case when no coherences are present in the system. In particular, (\ref{eq:mt_stoch}) and (\ref{eq:mt_gpcond}) hold also for the protocols described in \cite{Aberg2014, Vaccaro2018} which involve creating coherence between energy states of the system using initial coherence of the battery state. These protocols exploit full translational invariance of the battery: the unitary (which conserves the total energy of the system and battery, and acts trivially on the environment) allows for transitions only within the nearest energy levels of the battery. Even if the battery is bounded from below, by utilising the so called 'regenerative catalytic cycles', the protocol used by \cite{Aberg2014} quarantees that the ground state of the battery is never occupied. Therefore, it satisfies the condition (\ref{eq:mt5}) with  $k_{\rm min}=0$.     

\section{General results}

\subsection{Jarzynski equality for physical batteries}
The standard Jarzynski equality can be written as \cite{Jarzynski1997,Jarzynski1997a}:
\begin{align}
    \label{eq:mt9}
    \langle e^{\beta w} \rangle = \frac{Z_{\rm S'}}{Z_{\rm S}},
\end{align}
where $Z_{\rm S'}$ and $Z_{\rm S}$ are partition function associated with the final and initial Hamiltonian on the system $\rm S$. Recently, Alhambra \emph{et. al} in \cite{Alhambra2016} using the framework of thermal operations with the ideal weight derived a quantum analog of (\ref{eq:mt9}) relating fluctuations of work in an arbitrary thermodynamic process. In particular they showed that for states block-diagonal in the energy eigenbasis the following identity holds:
 \begin{align}
 \label{eq:mtje}
 \langle e^{\beta\, (w - f_s)} \rangle = Z_{\rm S'}, 
\end{align}
where the fine-grained free energy $f_i := E_i + \frac{1}{\beta} \log p_{\rm S}(i)$ and $w$ are random variables, $E_i$ is the energy of $i$-th level the system $\rm S$ and averaging is over $p(s,w)$, that is the probability that system $\rm S$ starts in level $\ket{s}_{\rm S}$  and performs work $w$. The above equality can be thought of as an extension of (\ref{eq:mt9}) to the case when the initial state of the system is out of equilibrium. When the initial state is thermal then (\ref{eq:mtje}) reduces to the standard Jarzynski equality (\ref{eq:mt9}).

The proof of (\ref{eq:mtje}) presented in \cite{Alhambra2016} relies crucially on the TI assumption imposed on thermal operations. In the following theorem we show that when translational symmetry is broken the equality (\ref{eq:mtje}) turns into a family of Jarzynski-like inequalities. The crucial difference here is that in the regime of broken translational symmetry (vacuum regime) there are no constraints of this form. At  first glance this may lead to the conclusion that Jarzynski equality in the form (\ref{eq:mtje}) can be violated arbitrarily well. However, in the following theorem we show that by looking at an analog of (\ref{eq:mtje}) conditioned on the battery level $k$, we can deduce nontrivial bounds which hold in the translationally invariant regime of our model. Moreover, in the next section we will show that these relations imply that the work fluctuations are constrained in such a way that the second law of thermodynamics still holds. 

Before we present the theorem let us first rewrite (\ref{eq:mtje}) in a form which explicitly demonstrates its dependence on the  initial battery state $\rho_{\rm W} = \sum_{k=0}^N p_{\rm W}(k) \dyad{\epsilon_k}_{\rm W}$, that is:

\begin{align}
\label{mt:exp_fluct}
     \langle e^{\beta(w - f_s)} \rangle\! =\! \sum_{k,k'} p_{\rm W}(k)   \langle e^{\beta(w_{kk'} - f_s)} \rangle_k,
\end{align}
where we denoted $w_{kk'} = \epsilon_{k'} - \epsilon_k$ and introduced:
\begin{align}
\langle e^{\beta( w_{kk'}- f_s)} \rangle_k := \\ \sum_{s, s', k'} &p_{\rm S}(s)\, r(s'k'|sk) e^{\beta(w_{kk'} - f_s)}. \nonumber
\end{align}
With this in mind we can present our first main theorem:
\begin{theorem}
\label{theorem1}
Let $\Gamma_{\rm SW}$ be a thermal operation acting on a harmonic oscillator battery and satisfying (\ref{mt:eti}). Then for all $k \geq k_{\emph{min}}$,
\begin{align}
    \label{mt:expkb}
     \langle e^{\beta(w_{kk'} - f_s)} \rangle_k \leq Z_{\rm S'}\left(1 + e^{ - \beta \delta_k} \right)\!,\!
\end{align}
where $\delta_k := \epsilon_k - \epsilon_{{\emph{min}}} + \delta$ is the energy difference between the state $\ket{\epsilon_k}_{\rm W}$ and the top of the vacuum regime $\ket{\epsilon_{\emph{min}} - \delta}_{\rm W}$. 
\end{theorem}
The theorem does not constrain fluctuations of work when the battery starts in the energy subspace without translational symmetry ($k < k_{\text{min}}$). In other words, if the symmetry (\ref{mt:eti}) is violated for some energy subspace and the battery starts in that subspace, the r.h.s. of (\ref{eq:mtje}) can be made arbitrarily large (and in fact depends on the dimension of the battery as can be seen from the proof of the theorem). This phenomenon did not occur for the ideal weight because in that case (\ref{mt:eti}) was satisfied for all energy levels (meaning that the battery had effectively no vacuum regime). When the initial state of the system is the equilibrium state then (\ref{mt:expkb}) leads to the following family of inequalities:
\begin{align}
    \forall \, k \geq k_{\text{min}}\quad\langle e^{\beta w_{kk'}} \rangle_k \leq \frac{Z_{\rm S'}}{Z_{\rm S}}\left(1 + e^{ - \beta \delta_k} \right),
\end{align}
Finally, the following example explicitly demonstrates that work can fluctuate arbitrarily for a valid thermal operation when the battery starts in the vacuum regime. 
\begin{example}[Deviation from the Jarzynski expression] 
\label{ex:2}
Consider a qubit system $\rm S$ and a thermal operation $\Gamma_{\rm SW}$ with a constant trivial Hamiltonian $H_{\rm S} = H_{\rm S'} = 0$ and described by its action on the basis states:
\begin{align}
    &\dyad{s}_{\rm S}\! \ot\!  \dyad{\epsilon_k}_{\rm W}\! \xrightarrow{}\! \dyad{0}_{\rm S} \!\ot\! \dyad{\epsilon_{k-1}}_{\rm W}\!, \!\\ 
     &\dyad{s}_{\rm S}\! \ot\!  \dyad{\epsilon_0}_{\rm W}\! \xrightarrow{}\! \dyad{1}_{\rm S} \!\ot\gamma_{\rm W},
\end{align}
where $\gamma_{\rm W}:= \sum_{k'=0}^{N}\! 2^{-k'-1}\dyad{\epsilon_{k'}}_{\rm W}$ and $k > 0$. Below, we will take the limit $N\rightarrow \infty$, which assures proper normalisation,  and with the energy gap $\delta = {\beta}^{-1} \log 2$ implies that the map is a valid thermal operation (its construction is discussed in more detail in Section \ref{Sec:Fluctuations}, subsection Physical Battery)). The process can be fully characterized by the transition probabilities:
\begin{align}
&k > 0: \qquad \quad r(0, \, k-1|s, \, k) = 1, \\
 &k = 0: \qquad \quad r(1, \, k'|s, \, 0) = 2^{-k'-1},
\end{align}
\noindent where all other transition probabilities are equal to zero. Notice that the process satisfies (\ref{eq:mt5}) for all $k > 0$, i.e. the vacuum regime is spanned just by one energy level $\dyad{\epsilon_0}_{\rm W}$. Let us now compute the averaged term from the Jarzynski equality (\ref{mt:exp_fluct}). 
We have:
\begin{align}
     \langle e^{\beta(w_{kk'} - f_s)} \rangle_k \! =\! & 
    \begin{cases}
    Z_{\rm S'} & \text{if } k > 0, \\ 
    Z_{\rm S'}(N\!+\!1) & \text{if } k = 0.
    \end{cases}
\end{align}
This means that the average (\ref{mt:exp_fluct}) goes to infinity (in the limit $N\rightarrow \infty$) whenever the initial occupation of the battery ground state is non-zero. 
\end{example}

We described a thermal operation for which Jarzynski equality does not provide a meaningful bound on work fluctuations when the battery initially occupies an energy eigenstate for which ETI does not hold. According to Theorem \ref{theorem1}, when the battery is initiated in such an eigenstate, these fluctuations may be higher than in the case of the ideal weight (for which standard Jarzynski equality holds ''everywhere''). Surprisingly, in such cases it is still possible to derive a form of the second law of thermodynamics with corrections depending on these occupations and the distance of the battery state to the vacuum regime. These correction terms quantify which part of the average energy change on the battery must be associated with heat rather than work. In the next section we present a general theorem which allows to quantitatively determine which part of the average energy change on the battery can be considered as a heat and which as a useful work.

\subsection{Second law of thermodynamics for physical batteries}
Thermodynamic work can be largely influenced by energy fluctuations in the system. Although both of them contribute to the change in the system's average energy, the work should be stored in an ordered form so that it can be later used for another transformation. At the same time heat is irreversibly dissipated and lost. In this section we describe the second main contribution of this paper which is a modified version of the second law of thermodynamics valid for batteries bounded from below. This allows to estimate how much of the average energy change on the battery can be associated with thermodynamic work when its spectrum is bounded from below.
\begin{theorem}
\label{mt:theorem2}
Let $\Gamma_{\rm SW}$ be a thermal operation satisfying (\ref{mt:eti}). Then for all battery states of the form:
\begin{align}
\rho_{\rm W} = \sum_{k=0}^{N} p_{\rm W}(k) \dyad{\epsilon_k}_{\rm W}
\end{align}
the average work $\langle w \rangle$ satisfies:
\begin{align}
    \label{mt:2nd_law_corr}
     \langle w \rangle \leq - \Delta F_{\rm S} + A_{\beta}(\rho_{\rm W}, \rho_{\rm S}) + B_{\beta}(\rho_{\rm W})
\end{align}
where: 
\begin{align}
    A_{\beta}(\rho_{\rm W}, \rho_{\rm S}) &:= \!  \sum_{\mathclap{k < k_{\emph{min}}}}\, p_{\rm W}(k)\!\left[E_{\rm S'}^{\emph{max}}\! -\! F(\rho_{\rm S})\! -\!  \eta_{\rm S}\frac{\partial \eta_k}{\partial \beta} \right] \\
    B_{\beta}(\rho_{\rm W}) &\!:=\! \frac{1}{\beta} \log\! \left[1\!+\!\!\! \sum_{k \geq k_{\emph{min}}}^N \!\!\! p_{\rm W}(k)\, e^{-\beta \delta_k} \right]
\end{align}
where $F(\rho_{\rm S}) = \tr[\rho_{\rm S} H_{\rm S}] - S(\rho_{\rm S})$ is the free energy, $S(\rho_{\rm S}) = - \tr \rho_{\rm S} \log \rho_{\rm S}$ is the von-Neumann entropy, $\Delta F_{\rm S} = F(\rho'_{\rm S'}) - F(\rho_{\rm S})$ is the change in the free energy of the system, $\eta_k := Z_{\rm W}\, e^{\beta \epsilon_k}$, $\eta_{\rm S} := Z_{\rm S} \, e^{\beta E_{\rm S}^{\emph{max}}}$, $E_{\rm S'}^{\emph{max}} = \max_s \widetilde{E}_{s}$ is the largest energy of the system $\rm S$ and $\delta_k := \epsilon_k - \epsilon_{\emph{min}} + \delta$.
\end{theorem}

Let us analyze the terms appearing in (\ref{mt:2nd_law_corr}). The upper bound for average work has two correction terms which both depend on the initial state of the battery. First of them is proportional to the occupation below the threshold energy $\epsilon_{\text{min}}$ and describes the contribution to the battery's average energy related to changing its entropy (using the battery in its vacuum regime spoils the battery). The second term decreases exponentially fast with the distance to the threshold $\epsilon_{\text{min}}$ and effectively vanishes when battery starts far away from it. In particular, if the battery operates sufficiently far from the threshold energy, that is if $\sum_{k<k^*} p_{\rm W}(k) \approx 0$ for some $k^* \gg k_{\text{min}}$, then both correction terms vanish and (\ref{mt:2nd_law_corr}) reduces to the ordinary form of the second law.

We stress here that the possible violation of the second law inequality, which can occur when the battery is initialized in the proximity of the vacuum regime, is an indication that in such case the average change of the energy of the battery can no longer be considered as a valid description of thermodynamic work. In this way we can interpret Theorem \ref{mt:theorem2} as a quantitative tool for determining the part of the average energy change on the battery that cannot be associated with thermodynamic work for arbitrary thermodynamic processes. In other words, Theorem \ref{mt:theorem2} allows to determine how far the battery must be initialized in order to interpret the energy change on the battery as a genuine thermodynamic work.

\begin{proof}[Proof of Theorem \ref{mt:theorem2} (sketch)]
\noindent The standard approach in deriving the second law of thermodynamics (\ref{eq:mt1}) is to start with the fluctuation theorem (\ref{eq:mtje}) and upper bound the average work using convexity of the exponential function. However, as we saw in Example \ref{ex:2}, if the battery has an energy subspace for which the map is not translationally invariant, the r.h.s of (\ref{eq:mtje}) can be made arbitrarily large for some battery states. Hence we need to modify the method to obtain informative bounds. We start by decomposing the average work $\langle w \rangle$ into two terms, each related to a different regime of the initial state of the battery:
\begin{align}
	\label{eq:mt7}
    \langle w \rangle = \sum_w p(w) \, w = \langle w \rangle_{\text{vac}} + \langle w \rangle_{\text{inv}},
\end{align}
where we labelled: 
\begin{align}
    \langle w \rangle_{\text{vac}} &:= \sum_w \sum_{k < k_{\text{min}}} p_{\rm W}(k)\, p(w|k) \, w\\ 
    \langle w \rangle_{\text{inv}} &:= \sum_w \sum_{k \geq k_{\text{min}}} p_{\rm W}(k)\, p(w|k)\,   w.
\end{align}
Our strategy is to independently bound both terms appearing in (\ref{eq:mt7}). Regarding the first term note that all we know is that the transition probabilities $\{r(s'k'|sk)\}$ come from a stochastic map (\ref{eq:mt_stoch}) which preserves the associated Gibbs state (\ref{eq:mt_gpcond}). In particular, this means  that they are all upper bounded by respective Gibbs factors. This means that for all input and output pairs $(s, k)$ and $(s', k')$ we can write $r(s'k'|sk) \leq e^{-\beta(\widetilde{E}_{s'} - E_s + w_{kk'}})$, which after some manipulation leads to:
\begin{align}
	\label{mt:v_vac}
    \langle w \rangle_{\text{vac}}  &\leq - \eta_{\rm S} \sum_{k < k_{\text{min}}}  p_{\rm W}(k) \frac{\partial \eta_k}{\partial \beta}
\end{align}
Consider now the second term from (\ref{eq:mt7}). Notice that now the sum runs over $k \geq k_{\text{min}}$ and so the assumptions of Theorem \ref{theorem1} are satisfied. Hence, using Theorem \ref{theorem1} and the convexity of the exponential function it can be shown that:
\begin{align}
    \label{mt:w_inv}
    \langle w \rangle_{\text{inv}} \leq & -\Delta F_S +B_{\beta}(\rho_W) \\ \nonumber &+  \left(\sum_{k= 0}^{k_{\text{min}} - 1}  p_{\rm W}(k)\right) \cdot (E_S^{\text{max}} - F(\rho_{\rm S})). 
\end{align}
The theorem follows by combining bounds (\ref{mt:v_vac}) and (\ref{mt:w_inv}). 
\end{proof}
 
\begin{figure}[h!]
\centering
\includegraphics[width=\linewidth]{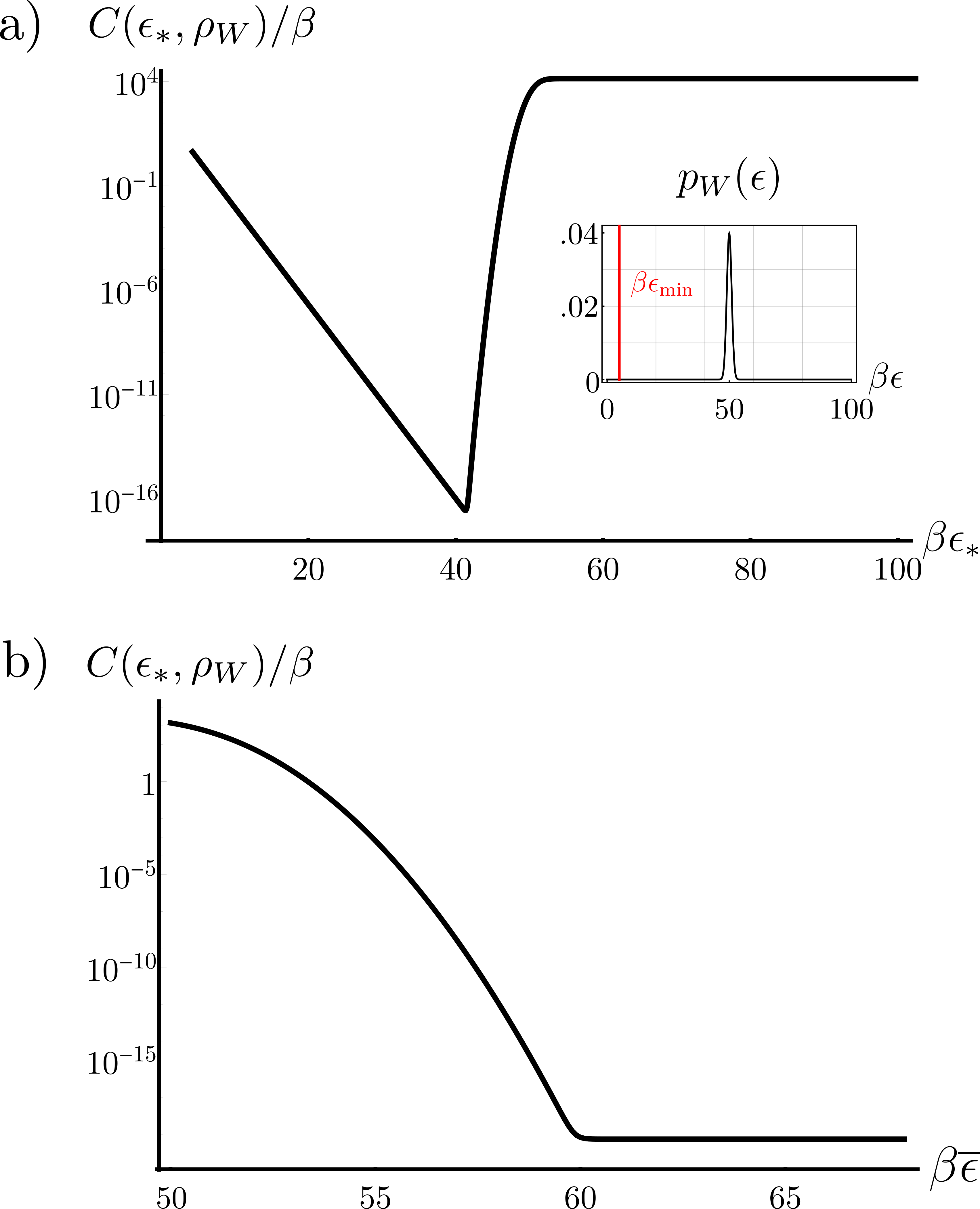}
\caption{ (a) Correction $C(\epsilon_*, \rho_{\rm W})$ for the bound on average work from Corollary \ref{cor1}. The initial state of the battery around the average energy $\beta \bar{\epsilon} = 50$ has a Gaussian profile $p_{\rm W}(\epsilon) \propto e^{\beta^2( \epsilon - \bar{\epsilon})^2 / 2}$. The energy spacing was set to $\beta \delta = 0.1$ and vacuum regime to $\beta \epsilon_{\text{min}}= 5$. System $\rm S$ was chosen to be a qubit with Hamiltonian $H_{\rm S} = 0$. Notice that in this case the battery's state is fixed. Hence by increasing $\epsilon_*$ we reach a point where $p(\epsilon < \epsilon_*)$ is no longer small (which happens around $\beta \epsilon_* = 40$) and so the correction term blows-up. (b) Correction $C(\epsilon_*, \rho_{\rm W})$ for the same setting but with fixed $\epsilon_* = 10 \times \epsilon_{\text{min}}$ plotted as a function of the average energy $\beta \bar{\epsilon}$ of the Gaussian profile. The kink appears as a consequence of fixing $\epsilon_*$ in Eq. (\ref{eq:c_corr}), while the first term vanishes exponentially with $\bar{\epsilon}$.}
\label{fig0}
\end{figure}

To illustrate Theorem \ref{mt:theorem2} more clearly consider an infinite-dimensional battery $(N\rightarrow\infty)$ with initial energy population $p(\epsilon<\epsilon_*)$ below some energetic cut-off $\epsilon_{*} > \epsilon_{\rm min}$, that is:
\begin{align}
    \label{eq:cut_off}
    p(\epsilon < \epsilon_*) = \sum_{k: \, \epsilon_k \leq \epsilon_*} \bra{\epsilon_k}  \rho_{\rm W} \ket{\epsilon_k}_{\rm W}.
\end{align}
The parameter $\epsilon_*$ will serve us to describe the initial state of the battery. In the following corollary we present a simplified (though slightly looser) bound derived from Theorem \ref{mt:theorem2}.  As we shall see, in this bound the dependence of the correction term on initial state of the battery is much simpler and in fact can be described by using just the function $p(\epsilon < \epsilon^*)$.
\begin{corollary} \label{cor1}
For any thermal operation $\Gamma_{\rm SW}$ acting on a harmonic oscillator battery $\rm W$ with threshold energy $\epsilon_{\emph{min}}$ and initial state $\rho_{\rm W}$, satisfying (\ref{eq:cut_off}), we have:
 \begin{equation}
 \langle w \rangle \leq - \Delta F_{\rm S} + \frac{1}{\beta} C(\epsilon_{*}, \rho_{\rm W}),
 \end{equation}
with
\begin{flalign}
 C(\epsilon_{*}, \rho_{\rm S}) =& \nonumber \,  p(\epsilon<\epsilon_{*})\Big[ c_{\rm S}\, h(\beta,\delta, \epsilon_{\rm min}) +\!\log c_{\rm S}\Big] \\ &+ c_{\rm S}  e^{-\beta (\epsilon_{*}-\epsilon_{\emph{min}})},
 \label{eq:c_corr}
 \end{flalign} 
where $h(\beta,\delta,\epsilon_{\rm min}) \!:=\! e^{-\beta\delta}[1\!+\! {\beta \delta e^{\beta \epsilon_{\emph{min}}}}{(1\!-\!e^{-\beta\delta})^{-2}}]\!$ and the constant $c_{\rm S} := d_{\rm S} e^{\beta E_{\rm S'}^{\emph{max}}}$ depends only on the properties of system $\rm S$.
\end{corollary}

Notice that for battery states concentrated far from the vacuum regime the term $C(\epsilon_*, \rho_W)$ vanishes exponentially fast in the low energy regime (see Fig. \ref{fig0}). Proof of the corollary can be found in Appendix \ref{app_proof_thm2}. 

\subsection{Recovering deterministic work}
All operations allowed in the framework involving the ideal weight can by carried out using the harmonic oscillator battery with a bounded spectrum from below, as long as we use the battery sufficiently high above the ground state (see e.g. \cite{Aberg2014,Aberg2018}). However, the converse statement is not true. In this section we show that there are thermal operations which can only be accomplished using batteries bounded from below. That means that the sets of operations generated by these two models are not equivalent, i.e. physical batteries are \emph{strictly} more powerful than ideal batteries. 

In this section we describe a method of extending arbitrary thermal operations defined on a wit to thermal operations acting on a harmonic oscillator battery. First of all, this construction shows that it is possible to recover the notion of deterministic (i.e. fluctuation-free) work for batteries with a ground state whenever they operate above the regime of broken translational symmetry. More intuitively this means that by properly breaking translational symmetry one can minimize to zero the fluctuations of work, while still satisfying the second law of thermodynamics.

Secondly, \emph{any} thermal operation acting on a wit can be extended using the construction provided below. In the Appendix we show that any thermal operation arising from this construction satisfies the second law in the sense of Theorem \ref{theorem1}. In this way, the maps which leave the wit in a mixed state can still be ``salvaged'' and lead to the average work which approximately obeys the second law of thermodynamics and for which the size of violations can be easily controlled.

We will now describe the construction. Let $\,\Gamma_{wit}$  be an arbitrary thermal operation acting on $\rm S$ and a two-level battery (wit) performing the transformation:
\begin{align}
    &\rho_{\rm S}\ot \dyad{0}_{\rm W} \rightarrow \\ \nonumber & \qquad \qquad\qquad \mathcal{R}_{00}(\rho_{\rm S}) \!\ot\! \dyad{0}_{\rm W}\! +\! \mathcal{R}_{01}(\rho_{\rm S})\! \ot\! \dyad{1}_{\rm W}\! \\
    &\rho_{\rm S}\ot \dyad{1}_{\rm W} \rightarrow \\ \nonumber & \qquad \qquad\qquad \mathcal{R}_{10}(\rho_{\rm S}) \!\ot\! \dyad{0}_{\rm W}\! +\! \mathcal{R}_{11}(\rho_{\rm S})\! \ot\! \dyad{1}_{\rm W}
\end{align}
Note that full information about $\Gamma_{wit}$ is contained in the set of subchannels $\{\mathcal{R}_{kk'}\}$ which we will refer to as \emph{battery subchannels}. Let $\{ \mathcal{R}_{kk'} \} = \{\mathcal{R}_{00}, \mathcal{R}_{01}, \mathcal{R}_{10}, \mathcal{R}_{11}\}$ be (arbitrary) battery subchannels associated with the transformation on the wit. The transition probabilities $\{r(s'k'|sk)\}$ can be extracted from $\{\mathcal{R}_{kk'}\}$ via:
\begin{align}
    \label{mt:rRcorr}
    r(s'k'|sk) = \tr \left[\dyad{s'} \mathcal{R}_{kk'} [\dyad{s}]\right]
\end{align}
We will use the map $\Gamma_{wit}$ as a primitive in constructing a family of thermal operations acting on a harmonic oscillator battery. We define a transformation $\Gamma_{osc}$ acting on $\rm S$ and a harmonic oscillator battery $\rm W$ in the following way:
\begin{construction}
\label{con1}
The action of  $\Gamma_{osc}$ is given by:
\begin{align}
    &\rho_{\rm S} \ot \dyad{\epsilon_0}_{\rm W} \rightarrow \sum_{i = 0}^{\infty} \mathcal{R}_{00} \mathcal{R}_{01}^{i}(\rho_{\rm S}) \ot \dyad{\epsilon_i}_{\rm W},  \\  
    &\rho_{\rm S} \ot \dyad{\epsilon_k}_{\rm W} \rightarrow \mathcal{R}_{10}(\rho_{\rm S}) \ot \dyad{\epsilon_{k-1}}_{\rm W} + \\ \nonumber 
    & \qquad \qquad  \,  \sum_{i = 0}^{\infty} \mathcal{R}_{00} \mathcal{R}_{01}^i \mathcal{R}_{11}(\rho_{\rm S}) \ot \dyad{\epsilon_{k+i}}_{\rm W}, 
\end{align}
for all $k > 0$.
\end{construction}
\noindent
\begin{figure}[h]
\centering
\includegraphics[width=\linewidth]{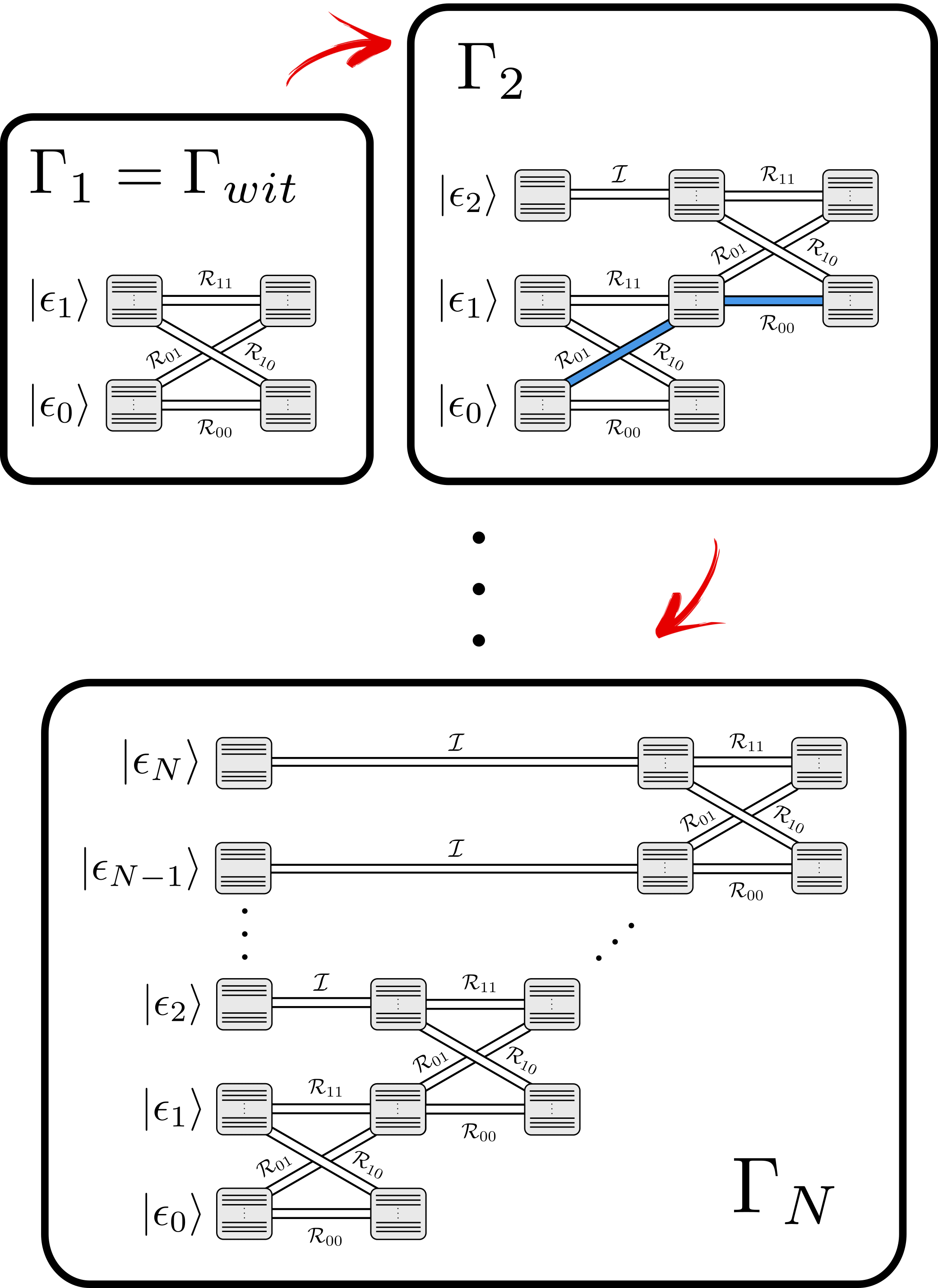}
\caption{Graphical explanation of the construction extending thermal operation $\Gamma_{wit}$ defined on a qubit battery $W_{wit}$ into a map $\Gamma_N$ acting on an $(N+1)$-level harmonic oscillator. When the battery has an infinite spectrum the construction yields a map $\Gamma_{osc}[\cdot] := \lim_{N \rightarrow \infty } \Gamma_{N} [\cdot]$. In the Appendix we show that for every $N \in \mathbb{N}$ channel $\Gamma_N$ is a thermal operation and in the limit $N \rightarrow \infty$ becomes effectively translationally invariant with threshold energy $\epsilon_{\text{min}} = \delta$. Blue color corresponds to an exemplary battery subchannel for a $3-$level battery, that is $\mathcal{R}_{01}^{(2)}\left[\cdot \right] = \mathcal{R}_{00} \mathcal{R}_{01} \left[\cdot \right]$. Other battery sub-channels can be determined in an analogous way.} 
\label{mfig2}
\end{figure}
In Appendix \ref{app_prop_ext_maps} we show that any map arising from this construction is a valid thermal operation. Moreover, the map additionally satisfies the ETI property (\ref{mt:eti}) with $\epsilon_{\text{min}} = \epsilon_1$. This implies that the assumptions of Theorem \ref{mt:theorem2} are met and $\Gamma_{osc}$ satisfies the second law in the form (\ref{mt:2nd_law_corr}).  In fact, with a small modification we can define the maps $\Gamma_{osc}$ also for a battery with a finite number of energy levels. To simplify presentation we postpone the details of this construction to the Appendix.

Finally, we apply the above construction to a  primitive map which uses work stored in the battery to form an arbitrary (energy-incoherent) quantum state out of a thermal state. As a result we obtain a method of recovering deterministic work for transformations acting on a harmonic oscillator battery with a spectrum bounded from below. In this sense we can think about the harmonic oscillator battery as owning both benefits of the previous two battery models: it allows for studying deterministic work which is an important concept in the resource-theoretic approach (property of the wit) while satisfying the second law of thermodynamics (property of the ideal weight).    
\begin{theorem}
\label{mt:theorem3a}
Let $\,\Gamma_{\text{wit}}$  be a thermal operation acting on system $\rm S$ and a two-level battery with Hamiltonian $H_{\rm W} = \delta \, \dyad{1}_{\rm W}$ with $\delta \geq 0$ and performing the transformation:
\begin{align}
\label{eq:tran_form1}
\Gamma_{wit}\left[\rho_{\rm S} \otimes \dyad{1}_{\rm W}\right] = \sigma_{\rm S} \otimes \dyad{0}_{\rm W},
\end{align}
for some {energy-incoherent} states $\rho_{\rm S}$ and $\sigma_{\rm S}$. Then there is a thermal operation $\,\Gamma_{osc}$ acting on $S$ and an infinite harmonic oscillator battery which performs the map (\ref{eq:mt5}), i.e.:
\begin{align}
\label{eq:tran_form2}
\Gamma_{osc} \left[\rho_{\rm S} \otimes \dyad{\epsilon_k}_{\rm W} \right] = \sigma_{\rm S} \otimes \dyad{\epsilon_k-\delta}_{\rm W},
\end{align} 
for all $k > 0$. The average work associated with this transformation is given by:
\begin{align}
    \langle w \rangle = -\delta = D_{\emph{max}}(\rho_{\rm S}||\tau_{\rm S}) - D_{\emph{max}}(\sigma_{\rm S}||\tau_{\rm S}),
\end{align}
where $D_{\emph{max}}(\rho||\sigma) = \log \min \{\lambda: \rho \leq \lambda \sigma \}$.
\end{theorem}

\begin{proof}
Notice that (\ref{eq:tran_form1}) corresponds to a set of battery subchannels with $\mathcal{R}_{11}[\rho_{\rm S}] = 0$ and $\mathcal{R}_{10}[\rho_{\rm S}] = \sigma_{\rm S}$. Hence using Construction \ref{con1} leads to a map $\Gamma_{osc}$ which performs (\ref{eq:tran_form2}) for all $k > 0$. 
\end{proof}

If the map (\ref{eq:tran_form1}) is not exact but holds up to some $\epsilon$ error on the wit then deterministic work will have an additional contribution surpassing the free energy change, as specified by the correction terms in Theorem \ref{mt:theorem2}. However, this contribution can be suppressed by increasing the distance of the initial state of the battery to the vacuum regime. Thus we can start in an arbitrary (e.g. mixed) state of the battery and due to Theorem \ref{mt:theorem2} the approximate version of the second law of thermodynamics will still be satisfied. In this way the notion of deterministic work can be recovered when the energy spectrum of the battery is bounded from below.

Before we finish this section let us briefly note that the construction we present here allows one to recover proper thermodynamic work (i.e. satisfying the second law of thermodynamics) for all primitive maps on the wit. In particular, this means that also the thermalizing map from Example \ref{ex1}, when extended to a harmonic oscillator battery, no longer leads to violations of the second law whenever the battery is initialized above the vacuum state. 

\section{Landauer erasure and measures of work fluctuations}\label{Sec:Fluctuations}
In this section we examine the main differences between the ideal weight and the harmonic oscillator battery in the low-energy regime, that is when population of the vacuum state cannot be ignored. Naturally we expect that this would lead to certain limitations which do not occur for ideal weight battery. A question arises: what are the consequences of these limitations, and most importantly, what are their implications for general thermodynamic protocols? Here we provide a partial answer to this problem. Our main goal is to explore how the presence of the ground state is reflected in the fluctuations of thermodynamic work.

In what follows we focus our attention on the paradigmatic process of Landauer erasure and implement it using two different battery models: the ideal weight and the harmonic oscillator. We will compare the minimal size of fluctuations that are necessary for these two battery models to perform Landauer erasure with certain fidelity.

Before going into the details, a remark has to be made about possible measures of fluctuations. The very notion of a process being deterministic is always related to the fact that, for this process, a selected measure of fluctuations of work distribution takes zero value. For the sake of comparison with earlier results for translationally invariant batteries \cite{Richens2016}, we start with a measure of fluctuations which is based on the distance to the average work: 
\begin{align}
F_1[w]=\max_{w:p(w)\neq 0} |w - \langle w \rangle|    
\end{align}
Thanks to Example $2$ we see that, when using an ideal weight battery, it is not possible to obtain a non-trivial transformation via a deterministic process with respect to this measure. However, in Section \ref{sec_4A} we show that this can be achieved when the same transformation is performed using a harmonic oscillator battery with unoccupied vacuum state. 
Later in Section \ref{sec_4B} we compare the differences between the two battery models when fluctuations are quantified by the variance of the work distribution:
\begin{align}
F_{2}[p(w)]=\int \text{d} w\, p(w) (w - \langle w \rangle)^2.    
\end{align}
Finally in Section \ref{sec_4C} we investigate a general quantifier of fluctuations:
\begin{align}
    F[p(w)]=\int \text{d} w\, p(w) f(w - \langle w \rangle),
\end{align}
where $f(x)$ is an arbitrary real function satisfying $f(0) = 0$. 

\subsection{Landauer erasure: fluctuations measured via $F_{1}[p(w)]$.}
\label{sec_4A}
Consider a qubit $\rm S$ with a constant Hamiltonian $H_{\rm S} = H_{\rm S'} = 0$ prepared in the Gibbs state $\tau_{\rm S} = \frac{1}{2} \mathbb{1}_{\rm S}$. Suppose that we also have an imperfect machine which is able to erase qubits with some small failure probability $\varepsilon$ and which can tolerate fluctuations of work up to a certain value $c$. We therefore demand that a bound on the distance to the average energy is satisfied for all registered values of $w$:
\begin{align}
    \label{eq:bnd_fluc}
    |w - \langle w \rangle| \leq c.
\end{align}
Our goal is to check if putting constraints on the allowed fluctuations of work of the form (\ref{eq:bnd_fluc}) can limit the ultimate precision of erasure that can be achieved in the process. Let us assume that our machine maps the qubit with probability $1-\varepsilon$ to the state $\ket{0}_{\rm{S}}$ and with probability $\varepsilon$ fails and outputs the orthogonal state $\ket{1}_{\rm S}$. The action of the machine on $\rm{S}$ can be described by the effective transformation:
\begin{align}
\label{m_eq:3}
\Gamma_{\rm{S}} \left[\tau_{\rm S}\right] = \rho_{\rm S}(\varepsilon),
\end{align}
where $\rho_{\rm S}(\varepsilon) = (1-\varepsilon) \dyad{0}_{\rm S} + \varepsilon \dyad{1}_{\rm S}$ and $\Gamma_{\rm S} = \tr_{\rm W} \Gamma_{\rm SW}$ is a thermal operation reduced to system $\rm S$. We will implement this effective transformation using two different battery models: the ideal weight and the harmonic oscillator. For the ideal weight we will look for an optimal transformation whereas for the harmonic oscillator we will apply Construction \ref{con1} to extend the Landauer erasure map on the wit to the harmonic oscillator battery. Our goal here is to find the minimal $\varepsilon$ which can be achieved when work fluctuations are bounded by $c$ according to (\ref{eq:bnd_fluc}).

\textit{\textbf{Ideal weight}.}
Let $\Gamma_{\rm SW} = \Gamma_{weight}$ be the thermal operation that performs (\ref{m_eq:3}) using the ideal weight battery. Since any such operation commutes with the shift operator on the weight (\ref{eq:com_rel}) the transformation will be independent on the initial state of the battery. Hence, without loss of generality, we can choose the initial battery state to be $\rho_{\rm W} = \dyad{0}_{\rm W}$. The necessary and sufficient conditions for the existence of a thermal operation realizing (\ref{m_eq:3}) in this case are given by (\ref{mt:7}). In our particular example (fully-degenerate Hamiltonian) they reduce to:
\begin{align}
    \label{eq:ifonlyif}
 \forall \, s' \qquad \sum_{s, w} p(s', w|s) e^{\beta w} &= 1,
\end{align}
Given a process $\Gamma_{weight}$ described by $p(s', w|s)$ and satisfying (\ref{eq:ifonlyif}), the work generated in a particular state transition $\ket{s}_{\rm S} \rightarrow \ket{s'}_{\rm S}$ has probability distribution given by:
\begin{align}
    \label{work:flucts}
    p(w|s, s') = \frac{p(s', w|s)}{p(s'|s)}, 
\end{align}
where $p(s'|s) = \sum_w p(s', w|s)$ is the probability of a given transition on the system. In general work as described by (\ref{work:flucts}) fluctuates around the mean value $\langle w \rangle$. The authors of \cite{Richens2016} (Theorem 2, Supplementary notes) showed that the process which minimizes work fluctuations (\ref{eq:bnd_fluc}) without increasing its average value $\langle w \rangle$ satisfies:
\begin{align}
    p(w|s,s') = \delta(w - w_{ss'}).
\end{align}
In other words, the optimal process which minimizes work fluctuations has a fixed value of work $w_{ss'}$ for any given state transition $\ket{s}_{\rm S} \rightarrow \ket{s'}_{\rm S}$. 
Sometimes in literature such processes are called ``deterministic'' to reflect the fact that for each transition there is a well-defined work cost. In this work we will use the term ''deterministic'' only with respect to the distribution $p(w)$ and only when a specific measure of fluctuations becomes zero.  For simplicity of presentation we will assume that $w_{00} = w_{10} = w_0$ and $w_{01} = w_{11} = w_1$, i.e. the random variable $w$ takes the value $w_{0}$ when we erase the state and $w_{1}$ when we fail and output an orthogonal state. Applying transformation $\Gamma_{weight}$ to the initial state of the qubit and the ideal weight battery leads to the joint state:
\begin{align}\label{erasure}
\Gamma_{weight} \Big[\tau_{\rm S} \otimes  &\dyad{0}_{\rm W} \Big] = 
\\ &=\sum_{s,s', w} p(s', w|s) \dyad{s'}_{\rm S} \otimes \dyad{w}_{\rm W} \nonumber
\\ &=\sum_{s,s'} p(s'|s) \dyad{s'}_{\rm S} \otimes \dyad{w_{s'}}_{\rm W}, \nonumber
\end{align}
where the probabilities $p(s'|s)$ are chosen such that the transformation reduced to $\rm S$ correctly  reproduces (\ref{m_eq:3}).
The action of $\Gamma_{weight}$ is summarized in Tab. \ref{mt:tab}.
\begin{table}[!htbp]
  \begin{tabularx}{0.5\textwidth}{@{}lllcl@{}}
    \hline\hline
    & $ \ket{\mathbf{s}} \rightarrow \ket{\mathbf{s'}}$ \qquad \qquad & $\mathbf{p(s'|s)}$ \qquad \qquad & \textbf{work} & \\
    & $\ket{0} \rightarrow \ket{0}$ & $1-\varepsilon$  & $w_{0}$ &  \\
    & $\ket{1} \rightarrow \ket{0}$ & $1-\varepsilon$  & $w_{0}$ &  \\
    & $\ket{0} \rightarrow \ket{1}$ & $\varepsilon$  & $w_{1}$ &  \\
    & $\ket{1} \rightarrow \ket{1}$ & $\varepsilon$  & $w_{1}$ &  \\
    \hline\hline
  \end{tabularx}
  \caption{The action of map $\Gamma_{weight}$ on system $\rm S$ with the associated work costs.}
  \label{mt:tab}
\end{table}
Notice that the conditions (\ref{eq:ifonlyif}) imply that the shifts $\{w_{s'}\}$ must necessarily satisfy:
\begin{align}
\label{app_eq20}
(s' = 0) \qquad \qquad  e^{\beta w_{0}}  &= \frac{1}{2(1-\varepsilon)},
\\
\label{app_eq21}
(s' = 1) \qquad \qquad e^{\beta w_{1}} &= \frac{1}{2 \varepsilon}.
\end{align} 
Using these we find that the random variable $w$ can take on one of the following values:
\begin{align}
\label{meq:10}
w_{s'}  =\left\{
                \begin{array}{ll}
                  -k_{\text{B}} T \log 2 - k_{\text{B}} T \log (1-\varepsilon) \, \quad  &{s'} = 0, \\ \\
                 - k_{\text{B}} T \log 2 - k_{\text{B}} T \log \varepsilon & {s'} = 1.
                \end{array}
              \right.
\end{align}
The average work during this process can be calculated as:
\begin{align}\label{eq:averagework}
    \langle w \rangle_{weight} &= \sum_w p(w) \, w \\
    &= \sum_{s,s'} p(s)\, p(s'|s) w_{ss'} \nonumber \\
    &= (1-\varepsilon) w_0 + \varepsilon \, w_1 \nonumber \\
    &= - k_{\text{B}} T \log 2 + k_{\text{B}} T  h(\varepsilon) \nonumber,
\end{align}
where $h(x) := -(1-x) \log(1-x) - x \log x$ is the binary entropy function. In the case when fluctuations of work are bounded by a constant $c$ as in (\ref{eq:bnd_fluc}) we have:
\begin{align}
    c &\geq \max_{w: \, p(w) > 0} |w - \langle w \rangle_{weight}| \nonumber \\
    &= \max_{s'} |w_{s'} - \langle w \rangle_{weight}| \nonumber \\
    &= k_{\text{B}} T \, |h(\varepsilon) + \log \varepsilon|\nonumber \\
    &\geq - k_{\text{B}} T (\log \varepsilon +\log 2).
\end{align}
This limits the range of $\varepsilon$ which can be achieved when the fluctuations are constrained by $c$, i.e.:
\begin{align}
    \label{eq:eps_bnd1}
    \varepsilon \geq \frac{1}{2} e^{-c/k_{\text{B}} T}.
\end{align}
If we now take the error $\varepsilon \rightarrow 0$ the average work $\langle w \rangle$ will approach a finite value of $\langle w \rangle = - k_{\text{B}} T \log 2$. However, this can be accomplished only when access to high energies is provided, so that fluctuations of work can be unconstrained, and $c\rightarrow \infty$. 

\textit{\textbf{Harmonic oscillator battery}.}
Let us consider again the process from (\ref{m_eq:3}), but now implemented using a harmonic oscillator battery. 
 In order to construct the desired thermal operation we start with a primary process acting on the wit battery with energy separation $\delta$, i.e.:
\begin{align}
\label{eq_eq36}
\Gamma_{wit} \left[ \tau_{\rm S} \otimes \dyad{1}_{\rm W} \right]  = \rho_{\rm S}(\varepsilon) \otimes \dyad{0}_{\rm W}.
\end{align}
Using standard methods (e.g. thermo-majorization curves) we can determine the minimal value of $\delta$ for which (\ref{eq_eq36}) is a valid thermal operation, i.e. $\delta = - k_{\rm B} T \cdot D_{\text{max}}(\rho(\epsilon)||\tau) = k_{\text{B}} T \log 2 + k_{\text{B}} T \log (1-\varepsilon)$. Recall that the action of any thermal operation $\Gamma_{wit}$ can be written as:
\begin{align}
\Gamma_{wit} \Big[ \rho_{\rm S}  \otimes \rho_{\rm W} \Big] = &\Big(\mathcal{R}_{00}\left[\rho_{\rm S}\right]\! +\! \nonumber \mathcal{R}_{10}\left[\rho_{\rm S}\right] \Big) \otimes \dyad{0}_{W}+ \\  &\Big(\mathcal{R}_{01}\left[\rho_{\rm S}\right]\! +\! \mathcal{R}_{11}\left[\rho_{\rm S}\right]\Big) \otimes \dyad{1}_{W}.
\end{align}
For diagonal input states we can encode the action of subchannels $\{\mathcal{R}_{kk'}\}$ using a set of substochastic matrices $\{R_{kk'}\}$ acting on the diagonals of respective states. Let us denote the vector of initial probabilities of system $S$ with $\mathbf{x}= \text{diag}(\rho_{\rm S})$. A simple analysis shows that the action of $\Gamma_{wit}$ on diagonal states can be expressed using the set of matrices:
\begin{align}
R_{00} &= \begin{bmatrix}\label{mat1}
0 & 0 \\
\frac{1-2 \varepsilon}{2(1-\varepsilon)} & \frac{1-2 \varepsilon}{2(1-\varepsilon)}
\end{bmatrix}, \quad 
R_{10} = \begin{bmatrix}
1-\varepsilon & 1-\varepsilon \\
\varepsilon & \varepsilon
\end{bmatrix}, \\ 
R_{01} &= \begin{bmatrix}
\frac{1}{2(1-\varepsilon)} & 0 \\
0 & \frac{1}{2(1-\varepsilon)}
\end{bmatrix},\quad
R_{11} \label{mat2}= \begin{bmatrix}
0 & 0 \\
0 & 0
\end{bmatrix}.
\end{align}
In the above we chose $R_{00}$ and $R_{01}$ such that $\Gamma_{wit}$ preserves the Gibbs-state and hence is a valid thermal operation.

Using Construction \ref{con1} we can now extend $\Gamma_{wit}$ to a thermal operation $\Gamma_{osc}$ acting on a harmonic oscillator battery. The matrices $\{\widetilde{R}_{kk'}\}$ describing this new process are given by:
\newline \newline
\textbf{For $k = 0: $} 
\begin{align}
 {\widetilde{R}}_{kk'} &=  {R}_{00}\, {R}_{01}^{k'} \\ \nonumber
    &= \begin{bmatrix}
        0 & 0 \\
        \frac{1-2 \varepsilon}{[2(1-\varepsilon)]^{k'+1}} & \frac{1-2 \varepsilon}{[2(1-\varepsilon)]^{k'+1}}
        \end{bmatrix}   
\end{align}
\textbf{For $k > 0 : $}
\begin{align}
 {\widetilde{R}}_{kk'} = 
    \begin{cases}
    {R}_{10}, & \text{if $k' = k -1$}, \\
    0, & \text{otherwise}.
   \end{cases}  
\end{align}
We will now consider two different cases: an ideal one in which the vacuum state of the battery is not occupied and the more physical one in which we assume a small (but nonzero) population in the vacuum state.

\noindent \textit{\textbf{Harmonic oscillator battery: vacuum not occupied}.} 
The transformation constructed using matrices $\{\widetilde{R}_{kk'}\}$ in Construction \ref{con1} has the same action for all battery states above the ground state, that is:
\begin{align}
\label{eq:oscmap}
\Gamma_{osc} \left[ \tau_{\rm S} \otimes \dyad{\epsilon_k}_{\rm W} \right] = \rho_{\rm S}(\varepsilon) \otimes \dyad{\epsilon_{k}-\delta}_{\rm W}, 
\end{align} 
This holds for all initial states of the battery above the ground state, i.e. for all $\epsilon_k > \epsilon_{\text{min}} = 0$. This means that any initial state $\rho_{\rm W}$ of the battery (with $\langle \epsilon_0 |\rho_{\rm W} |\epsilon_0 \rangle = 0)$ will lead to the same transformation. Furthermore, it can be shown (see Appendix \ref{app_prop_ext_maps} for the details) that for any primitive map $\Gamma_{wit}$ and any initial state of the battery above the ground state the average work associated with the extended map $\Gamma_{osc}$ can always be expressed as:
\begin{align}
    \langle w \rangle_{osc} &= \nonumber \delta \cdot \left[\mathbf{1}^T (\mathbb{1} - R_{01})^{-1} R_{11}\, \mathbf{x} - 1\right] \\
    & = -\delta,
\end{align}
where $\mathbf{1}^T = (1, 1, \ldots, 1)$ is the (horizontal) identity vector and $\mathbb{1}$ is the identity matrix.  From (\ref{eq:oscmap}) it is clear that whenever the battery starts above its ground state the shifts $\{w_{s'}\}$ are the same whenever erasure succeeds or fails, i.e. for all $s'$:
\begin{align}
    w_{s'} = -\delta = - k_{\text{B}} T \log 2 - k_{\text{B}} T \log(1-\epsilon).
\end{align}
This also implies that the random variable $w$ does not fluctuate. Indeed we have:
\begin{align}\label{fluct}
    c &\geq \max_{w: \, p(w) > 0} |w - \langle w \rangle_{\text{osc}}| \nonumber \\\nonumber
     &=  |\delta - \langle w \rangle_{\text{osc}}| \\
    &= 0.
\end{align}
In particular this means that $\varepsilon$ does not depend on $c$ and hence there is no fundamental limit on $\epsilon$ allowed by the transformation, i.e.:
\begin{align}
    \label{eq:eps_bnd2}
    \varepsilon \geq 0.
\end{align}
In this way a harmonic oscillator battery faithfully reproduces the amount of deterministic work needed to perform erasure for an arbitrary low error rate $\varepsilon$, a feature which the ideal weight could not perform as described in Example \ref{dwexample}.

\noindent \textit{\textbf{Harmonic oscillator battery: vacuum occupied}.}
While the above is a valid mathematical construction, in a physical situation the vacuum state will be inevitably occupied, and therefore one should expect the harmonic oscillator battery to show fluctuations of work. To quantify them, we assume that the vaccum state is occupied with probability $\gamma$. In this case the situation is different to the ideal ($\gamma = 0$) case. This is because applying the thermal map $\Gamma_{osc}$ when the battery is in its ground state thermalizes the battery, i.e. returns a state with all levels of the harmonic oscillator battery occupied. However, although the map brings the battery to a full-rank state, the probability of occupying higher energy levels decays exponentially with the energy. Hence the entropy change in the battery is always finite and, in particular, does not diverge in the limit of infinite dimension of the battery.

(though with probability decaying exponentially fast with the energy). Furthermore, the occupation of the vacuum induces a proportional error on the system. In order to find the action on the system and the total distribution of work $p(w)$ let us note that our construction assures translational invariance for all energy levels above energy ${\epsilon_{\rm min}} = \epsilon_0$. This means that without loss of generality we can prepare the battery in the initial state:
\begin{align}
    \rho_{\rm W} = (1-\gamma) \dyad{\epsilon_1}_{\rm W} + \gamma \dyad{\epsilon_0}_{\rm W},
\end{align}
where $\gamma \in [0, 1]$. Applying our Landauer erasure map leads to the following state on the system and the battery:
\begin{align}
    \!\!\!\Gamma_{osc}[\tau_{\rm S} \!\ot\! \rho_{\rm W}] \!&=\! (1-\gamma) \Gamma_{osc} [\tau_{\rm S}\ot  \dyad{\epsilon_1}_{\rm W}]  \\
    &\quad+ \gamma\, \Gamma_{osc}[\tau_{\rm S} \ot \dyad{\epsilon_0}_{\rm W}] \\
    &= (1-\gamma)\rho_{\rm S}(\epsilon) \ot \dyad{\epsilon_0}_{\rm W} \\ 
    &\quad +\sum_{i = 0}^{\infty} \mathcal{R}_{01}^i \mathcal{R}_{00}[\tau_{\rm S}] \ot \dyad{\epsilon_i}_{\rm W}.\!
\end{align}
The map $\Gamma_{osc}$ implements Landauer erasure on the system with the total error $\varepsilon_{\rm tot} = \varepsilon(1-\gamma) + \gamma$, where $\varepsilon$ is the parameter inherited from $\Gamma_{wit}$, based on which $\Gamma_{osc}$ is constructed. Therefore, even for $\varepsilon=0$ perfect erasure is not possible and ultimately depends also on the occupation of the vacuum state. 

In this case the protocol can only be implemented if we allow $c \rightarrow \infty$, i.e. work $w$  can take all possible values between $0$ and $\infty$ (although this happens with exponentially small probabilities). In this way, even though in the ideal case we could in principle always choose $\varepsilon = 0$ for any value of $c \geq 0$, when the vacuum state is populated we must allow for unbounded fluctuations of work in order to carry out the erasure perfectly.

To summarize, we see that using a battery with broken translational symmetry can lead to an arbitrarily good precision of Landauer erasure just as in the case of the ideal weight battery, but with work fluctuations reduced to zero. However, once we consider a more physical situation in which the battery's ground state is occupied, both models behave essentially in the same way. It is an interesting open question whether the different behavior certified for $\gamma = 0$ can lead to some physical advantages over the ideal weight (e.g. resulting from lifting the translational invariance constraint), or it is just a mathematical idealization, similar as in the case of wit. We leave this case open for future research. 

\subsection{Landauer erasure: fluctuations measured via $F_{2}[p(w)]$.} \label{sec_4B}
We saw above that for a non-zero occupation of the vacuum state, $\Gamma_{osc}$ returns a state with all levels of the harmonic oscillator battery occupied (though with probability decaying exponentially with the energy), while the average energy remains finite. Therefore, the previously used measure of fluctuations, $F_{1}[w]$, takes infinite values. In order to give a quantitative description of the functioning of physical batteries, we therefore switch to considerations of statistical moments of energy changes, as they take into account not only values of registered work, but also probabilities of obtaining it. 

Immediately we can draw some conclusions about the functioning of physical batteries (i.e. with ground states), even without invoking their translational invariance properties in some regime above the cut-off energy. Whenever the vacuum state (or the lowest energy eigenstate on which the process acts non-trivially) is occupied, stochasticity of the map implies that the work distribution has positive contributions resulting from the vacuum being populated (as population on this level cannot be mapped to lower levels). On the other hand, Landauer erasure is a type of transformation which requires work (hence in our notation it is associated with a negative average work). Therefore higher energy levels of the battery have to be initially occupied, in order to assure the dominant negative contribution to the average. As a consequence, the variance can only be reduced to zero when the vacuum is not populated. Moreover, as we show below in Theorem \ref{theorem4}, for a fixed occupation of the vacuum, the variance of Landauer erasure cannot even be brought down arbitrarily close to zero. This can be treated as another argument suggesting that the ideal weight may not be able to encapsulate important physical effects in the regimes where the occupation of the vacuum cannot be ignored. In other words, it allows for processes which cannot be realised in practise, i.e. when one only has access to physical batteries.

Before delving into quantitative analysis let us start with a simple fact. The below theorem provides a bound on the fluctuations of work which is valid for any battery with the ground state and any stochastic transformation on the system and the battery (i.e. it is not necessarily translationally-invariant or Gibbs-preserving).
\begin{theorem}
\label{theorem4}
Let $\rho_{\rm W}$ be an arbitrary state of the battery with a non-zero occupation $\gamma$ of the ground state $\ket{\epsilon_0}_{\rm W}$, i.e.:
\begin{align}
    \gamma = \langle \epsilon_0 |\rho_{\rm W} |\epsilon_0 \rangle.
\end{align}
The variance of the work distribution $p(w)$ arising from any thermodynamic protocol with $\langle w \rangle \leq 0$ is bounded by:
\begin{align}
    \emph{Var}[w] \geq \gamma \langle w \rangle^2.
\end{align}
\end{theorem}
\begin{proof}
The variance $\text{Var}[w]$ of a random variable $w$ distributed according to $p(w)$ is given by:
\begin{align}
    \text{Var}[w] = \sum_{w} p(w) (w - \langle w \rangle),
\end{align}
Let $s$ be the probability of performing work $w = 0$ when the battery is in the ground state. Using the assumption $\langle w \rangle \leq 0$ we can write:
\begin{align}
     \nonumber
     \text{Var}[w] &\geq p(w = 0) (0 - \langle w \rangle)^2 + \sum_{w > 0}p(w) (w-\langle w \rangle)^2 \\ \nonumber
     &\geq p(w = 0) \langle w \rangle^2 + \langle w \rangle^2 \sum_{w > 0} p(w) \\
     &\geq \gamma s \langle w \rangle^2 + \gamma (1-s) \langle w \rangle^2 
\end{align}
This proves the claim.
\end{proof}

\noindent \textit{\textbf{Harmonic oscillator battery}.}
In this section the average energetic cost of performing the transformation is measured in terms of the average work $\langle w\rangle_{osc}(\epsilon,\gamma)$, while the variance of the work distribution:
\begin{align}
    F_{2}[p(w)] &= \text{Var}_{osc}(\epsilon,\gamma) \nonumber \\
    &= \int \text{d} w\, p(w) (w - \langle w \rangle)^2
\end{align} 
will be used as a measure of fluctuations. For the map from Construction \ref{con1} based on (\ref{mat1}) and (\ref{mat2}), we can calculate the corresponding measures: 
\begin{equation}
\langle w\rangle_{osc}(\varepsilon,\gamma)/kT= - \delta \left(1-\frac{2 \gamma  (1-\varepsilon )}{1-2 \varepsilon }\right)
\end{equation}
and 
\begin{align}\label{var}
\frac{\text{Var}_{osc}(\varepsilon,\gamma)}{(kT)^2} = \gamma \delta^2 \, \frac{2 (1-\varepsilon ) (-2 \gamma  (1-\varepsilon )-2 \varepsilon +3) }{(1-2 \varepsilon )^2}.
\end{align}
where $\delta := \log 2(1-\varepsilon)$ is defined as before. It can be easily shown that, 
for a given value of $\varepsilon_{\text{tot}}=\varepsilon(1-\gamma)+\gamma$ both $\langle w\rangle_{osc}(\varepsilon,\gamma)$ and $\text{Var}_{osc}(\varepsilon,\gamma)$ are maximized for $\varepsilon=0$, $\gamma=\epsilon_{tot}$ (see also Fig. \ref{comparison}). As $\langle w\rangle_{osc}(\epsilon,\gamma)$ is always non-positive, its maximization corresponds to minimization of work spent on the erasure. However, the less work we spend on the erasure, the higher are its fluctuations, as measured by $F_2[p(w)]$. 

\noindent \textit{\textbf{Ideal weight battery}.}
These fluctuations on the harmonic oscillator battery should be compared with fluctuations of work registered on the ideal weight battery for an analogous process, i.e. a process implementing Landauer erasure with the same failure probability $\varepsilon_{\text{tot}}$. As a promising candidate for minimizing fluctuations in this regime, we take the process (\ref{erasure}) with work values specified in Table \ref{mt:tab} (which minimizes the measure $F_{1}[w]$).
Therefore, we denote the corresponding values of the average work and variance with subscript "$1$", i.e.:
\begin{equation}
\langle w\rangle_{1}(\varepsilon_{\rm tot})/kT=- \log 2 +h(\varepsilon_{\rm tot}),
\end{equation}
and 
\begin{align}
&\frac{\text{Var}_{1}(\varepsilon_{\rm tot})}{(kT)^2}  &\nonumber\\
=&\varepsilon_{\rm tot}\big(-\log (\varepsilon_{\rm tot})-h(\varepsilon_{tot})\big)^2&\nonumber\\
+&(1-\varepsilon_{\rm tot})\big(-\log (1-\varepsilon_{\rm tot})-h(\varepsilon_{\rm tot})\big)^2,
\end{align}
where $h(x)=-(1-x) \log(1-x) - x \log x$.
In the regime of small errors $\varepsilon_{\rm tot}$, fluctuations on the harmonic oscillator battery are smaller compared to the ones on the weight (see Fig. \ref{comparison}). As indicated above, this range grows with decreasing occupation of the vacuum state, in agreement with the description of the no-fluctuation case of the idealised harmonic oscillator battery presented in the previous section. We also observe that working with an oscillator battery leads to higher expenditures of work, as $\langle w\rangle_{osc}(0,\varepsilon_{\text{tot}})\leq \langle w\rangle_{1}(\epsilon_{\text{tot}})$, with equality for $\epsilon_{\text{tot}}=0$ and $\epsilon_{\text{tot}}=1$. 

Note however that it is not clear that the process (\ref{erasure}) minimizes the measure $F_2[w]$ for the ideal weight. In fact, the following probabilities define a process on the ideal weight which 
in the limit  $\lambda\rightarrow 0$ is deterministic with respect to $F_{2}[w]$:
\begin{align}\label{niedet}
    p(w|s' = 0) &= \delta(w-w_0), \\ p(w|s' = 1) &= (1 - \lambda) \delta(w-w_0) + \nonumber \\ &\hspace{14pt}\lambda \delta(w-w_1),   
\end{align}
where $\lambda \in [0, 1]$. Since in our case we have $p(s' = 0) = 1-\varepsilon$ and $p(s' = 1) = \varepsilon$ we can write the total distribution of work $p(w)$ as:
\begin{align}
    p(w) = (1-\lambda \varepsilon) \delta(w-w_0) + \lambda \varepsilon \delta(w-w_1).
\end{align}
Such a thermal operation exists if the Gibbs-stochasticity conditions are satisfied. This means that we must have:
\begin{align}
    \label{new_gsconds}
    e^{\beta w_0} &= \frac{1}{2(1-\varepsilon)}, \\
    (1-\lambda)e^{\beta w_0} + \lambda e^{\beta w_1} &= \frac{1}{2\varepsilon}.
\end{align}
Notice now that for an arbitrarily small $\varepsilon$ and an arbitrarily small $\lambda$ the work value $w_1$ can be always chosen large enough so that these two constraints are satisfied. In particular for any $\lambda,\varepsilon > 0$ the choice
\begin{align}
    w_1 = \frac{1}{\beta} \log\left[ \frac{1}{2 \lambda \epsilon} - \frac{1-\lambda}{2(1-\varepsilon)} \right]
\end{align}
leads to a legitimate thermal operation. Direct calculations in the limit $\lambda\rightarrow 0$ show:
\begin{align}
\langle w \rangle_2(\varepsilon) &= w_{0} = - k_{\rm B}T [\log 2 + \log(1-\varepsilon)], \\ 
\text{Var}_2(\varepsilon) &= 0. 
\end{align}
In this way by allowing for larger amounts of work which  occur with respectively smaller probabilities, we can recover deterministic work (with respect to $F_{2}[w]$). In fact, in the limit $\lambda\rightarrow 0$, process (\ref{niedet}) leads to the same average and variance of work as the process from Construction 1, with no vacuum state occupied. Note that, due to the convexity of the exponential function in (\ref{new_gsconds}) we have $\langle w\rangle_{2}\leq \langle w\rangle_{1}$, so here also a reduction of fluctuations is obtained at the cost of performing additional work.  

\begin{figure}[h!]
\centering
\includegraphics[width=\linewidth]{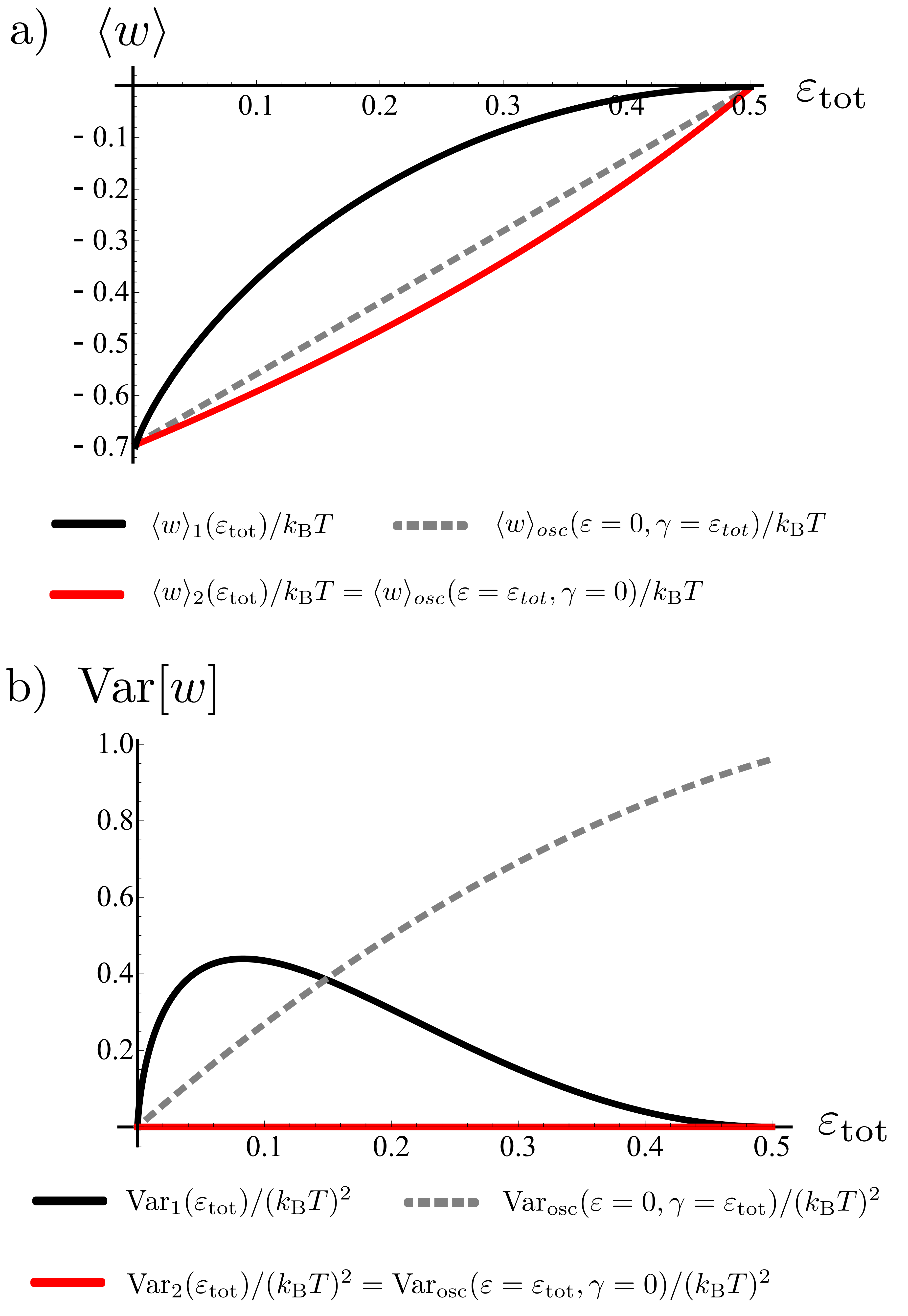}
\caption{\label{comparison}
Comparison between different processes performing Landauer erasure on ideal weight and physical battery, with error $\epsilon_{tot}$ on a system, stated in terms of (a) average work $\langle w\rangle$ and (b) variance of work distribution Var[$w$]. On physical battery, a process based on Construction \ref{con1} with (\ref{mat1}) and (\ref{mat2}), whose error $\varepsilon_{tot}$ stems entirely from vacuum occupation $\gamma$, is characterized by minimal work cost and maximal fluctuations indicated by the grey dashed lines. For the nonphysical case of the vacuum being not occupied, fluctuations can be reduced to zero, at the expense of increasing work cost (red lines). On the other hand, the same statistics can be achieved on the ideal weight for the process (\ref{niedet}) in the limit $\lambda\rightarrow 0$. For $\lambda=1$, solid black lines depict minimal possible work expense and corresponding variance of work distribution on the ideal weight.}
\end{figure}

The above shows that one can perform Landauer erasure using the ideal weight battery with deterministic work (with respect to variance as the measure of fluctuations). This, however, contrasts with the behaviour we observed for physical batteries which have a nonzero occupation of the ground state (see Theorem \ref{theorem4}). In the next section, focusing on processes which we chose here for illustrative purposes, we show the evidence suggesting that this discrepancy between battery models may be extended to a broader class of fluctuation measures.

\subsection{Landauer erasure: generalized model of fluctuations $F[w]$.}
\label{sec_4C}
Let us now consider a simple generalization of the quantifiers we studied in the previous sections, i.e. we will look at the following quantity:
\begin{align}
    \label{eq:f_new}
    F[w] := \int \text{d} w\, p(w) f(w - \langle w \rangle).
\end{align}
In the above $f(x)$ can be any function which satisfies $f(0) = 0$.  Intuitively, the above describes the (weighted) average work fluctuations that one is willing to tolerate in their protocol. In this way bounding (\ref{eq:f_new}) with some constant value $c$ provides a more general description of our willingness to tolerate fluctuations of work depending on their size. To gain some more intuition consider the following choice of the cost function $f(x)$:
\begin{align}
    f(x) = \begin{cases}
    0 &\text{if}\quad |x| \leq c, \\
    \infty &\text{else}.
    \end{cases} 
\end{align}
This corresponds to the situation of \emph{c-bounded work} which we considered in Section \ref{sec_4A}, when setting $c = 0$ leads to a deterministic work extraction, with respect to the measure $F_{1}[p(w)]$. On the other hand, choosing $f(x) = x^2$ leads to the variance of the work distribution (see Section \ref{sec_4B}).

We will now again consider the example of approximate Landauer erasure, but now expressing constraints in terms of the function $F$ defined in (\ref{eq:f_new}). Our goal here is similar as before: we would like to check if putting constraints on the allowed fluctuations of work can lead to any bounds on the ultimate precision of the Landauer erasure process. We will again look at the problem first from the perspective of using the ideal weight as the battery and second using the harmonic oscillator with the ground state.

\textit{\textbf{The ideal weight battery}.} 
Let us consider the approximate Landauer erasure process performed on the weight described by the process (\ref{niedet}).
As mentioned before, by allowing larger amounts of work $W_{1}$ occurring with respectively smaller probabilities $\varepsilon\lambda$, we can recover deterministic work (with respect to $F_{2}[p(w)]$). This analysis, however, does not take into account the size of fluctuations. We should expect that small fluctuations around the average $\langle w \rangle$ should not be as adversarial as the large ones, and similarly large fluctuations should be more costly to tolerate. Hence we instead put a constraint on the average cost of fluctuations:
\begin{align}
    \label{eq:const_c}
    F[w] \leq c.
\end{align}
To focus our attention let us choose an exponential cost function $f(x) = e^{|x|}-1$, i.e. we are going to treat small fluctuations as almost free and larger ones as significantly more costly. Any function with exponential or better decay in $x$ would lead to the same qualitative conclusions. In this case our constraint (\ref{eq:const_c}) becomes:  
\begin{align}
    (1-\lambda \varepsilon) e^{\beta(w_0 - \langle w \rangle)} + \lambda \varepsilon e^{\beta(w_1 - \langle w \rangle)} \leq c.
\end{align}
Since both $\lambda,\varepsilon > 0$ this also implies:
\begin{align}
     \label{eq:xyz1}
     \lambda \varepsilon e^{\beta w_1} \leq c \, e^{\beta \langle w \rangle} \leq c,
\end{align}
since $\langle w \rangle \leq 0$ which follows by applying Jensen's inequality to (\ref{new_gsconds}). Let us now look at the Gibbs-stochasticity conditions (\ref{new_gsconds}). By solving for $w_0$ we can express the second condition as:
\begin{align}
    \label{eq:xyz2}
    \lambda \varepsilon e^{\beta w_1} = \frac{1-2\varepsilon + \varepsilon \lambda} {2(1-\varepsilon)}\geq \frac{1 - 2 \varepsilon}{2(1-\varepsilon)}.
\end{align}
Combining (\ref{eq:xyz1}) with (\ref{eq:xyz2}) leads to the following bound on $\varepsilon$:
\begin{align}
    \varepsilon \geq \frac{1}{2} - \frac{c}{2(1-c)}.
\end{align}
Hence we have a non-trivial bound for all $c \in [0, \frac{1}{2})$. Importantly, this bound holds for all initial states of the weight battery and states that its impossible to obtain lower $\varepsilon$ for any possible thermal operation.

\textit{\textbf{Harmonic oscillator battery}.}
Let us now consider the same transformation but performed using the harmonic oscillator battery. In the ideal case when the vacuum state is not occupied the whole analysis can be performed just as in the previous section. However, the situation changes if we start with a non-zero population of the vacuum state. In that case $\Gamma_{osc}$ returns a state with all levels of the harmonic oscillator battery occupied (though with probability decaying exponentially fast with the energy). Furthermore, occupation of the vacuum induces a proportional error on the system. In order to find the action on the system and the total distribution of work $p(w)$ let us note that our construction assures translational invariance for all energy levels above $k = 0$. This means that without loss of generality we can prepare the battery in the following initial state:
\begin{align}
    \rho_{\rm W} = (1-\gamma) \dyad{\epsilon_1}_{\rm W} + \gamma \dyad{\epsilon_0}_{\rm W},
\end{align}
where $\gamma \in [0, 1/2]$. Applying our Landauer erasure map leads to the following final state on the system and the battery:
\begin{align}
    \Gamma_{osc}[\tau_{\rm S} \ot \rho_{\rm W}] &= \nonumber (1-\gamma) \Gamma_{osc} [\tau_{\rm S}\ot \dyad{\epsilon_1}_{\rm W}] + \\ &\nonumber \hspace{37pt} \gamma\, \Gamma_{osc}[\tau_{\rm S} \ot \dyad{\epsilon_0}_{\rm W}] \\ \nonumber
    &= (1-\gamma)\rho_{\rm S}(\varepsilon)_{\rm S} \ot \dyad{\epsilon_0}_{\rm W} + \\  & \qquad \gamma \sum_{i = 0}^{\infty} \mathcal{R}_{01}^i \mathcal{R}_{00}[\tau_{\rm S}] \ot \dyad{\epsilon_i}_{\rm W}.
\end{align}
Notice that we also have:
\begin{align}
    \mathcal{R}_{01}^i\mathcal{R}_{00}[\tau_{\rm S}] =
     \frac{1-2\varepsilon}{2^{i+1}(1-\varepsilon)^{i+1}} \dyad{1}_{\rm S}.
\end{align}
To simplify notation let us denote $\alpha_i := \tr\mathcal{R}_{01}^i\mathcal{R}_{00}[\tau_{\rm S}]$. The final state on the system becomes:
\begin{align}
    \rho_{\rm S}' &= \tr_{\rm W}  \Gamma_{osc}[\tau_{\rm S} \ot \rho_{\rm W}]  \nonumber \\ &= (1-\varepsilon_{\rm tot}) \dyad{0}_{\rm S}\! +\! \varepsilon_{\rm tot} \dyad{1}_{\rm S},
\end{align}
where $\varepsilon_{\rm tot} = \varepsilon + \gamma(1-\varepsilon)$ is the effective total error on the system. Hence it is now not only the channel imperfection but also occupation of the vacuum which bounds the precision of the Landauer erasure. The work distribution can be now computed as: 
\begin{align}
    p(w) &=  \left(1-\gamma\right) \delta(w+\delta) + \gamma \sum_{i=0}^{\infty} \alpha_{i} \delta(w-i\delta) 
\end{align}
The average work $\langle w \rangle_{osc}$ then becomes:
\begin{align}
    \langle w \rangle_{osc} = 
    -\delta \left(1 - \frac{2 \gamma(1-\varepsilon)}{1-2\varepsilon}\right)
\end{align}
In this case our average cost function (\ref{eq:const_c}) with exponential cost $f(x) = e^{|x|}$ becomes:  
\begin{align}
    F[p(w)] = (e^{2\beta \delta \gamma}-1) - \frac{\gamma}{2} e^{-2 \beta \delta \gamma}\left(1-\frac{e^{\beta \delta}}{2-e^{\beta \delta}}\right)
\end{align}
Notice that if we now choose the smallest possible $\delta = k_{\rm B}T \log [2(1-\varepsilon)]$ we obtain an expression diverging to infinity with $\varepsilon\rightarrow 0$. This also holds if we choose a bigger work value, i.e. $\delta = k_{\rm B}T \log 2$. 

We therefore see that, for a broad class of fluctuation measures $F[p(w)]$ (\ref{eq:f_new}) with $f(x)$ increasing at least exponentially, physical batteries with occupied vacuum lead to infinite deviations from deterministic work. This is in line with our observation that occupation of the ground state leads to fluctuations of work which can never be brought down arbitrarily close to zero. In this respect the two battery models (the ideal weight and the harmonic oscillator battery) lead to different thermodynamic predictions in the low energy regime, i.e. when the occupation of the ground state cannot be ignored.

\section{Summary and future work}
In classical thermodynamics the existence of a perfect “work receiver” is a standard and very often implicitly accepted assumption. This is also the case for many contemporary results in quantum and stochastic thermodynamics (e.g. fluctuation theorems) where work is described by energy changes on the system. In this work we argued towards breaking this paradigmatic assumption and studying a truly physical (i.e. \emph{explicit} and \emph{bounded}) work storage device, a harmonic oscillator battery.

We have observed that the effects related to the ground state of the work-storage device manifest in the Jarzynski fluctuation theorem. We determined the estimates on fluctuations of energy changes of the battery which replace the standard fluctuation theorem for the case when the battery is bounded from below.

Furthermore, using these estimates we derived corrections to the second law of thermodynamics and showed that they vanish exponentially fast with the distance of the initial state of the battery to the bottom of its spectrum. These corrections allow to determine when contributions to the average energy of the battery can be treated as average work, hence providing a quantitative tool to differentiate when the battery operates accurately and when it acts as an entropy sink. 

Notably, mathematical forms of our bounds on work fluctuactions and average work remain valid even when coherence is present in the system, though terms which in the incoherent case refer to free energy loose this association. As a consequence, the bounds can no longer be used as tools for assesment of physicality of work. It would be interesting to see if there is a positive interplay between quantum effects resulting from the existance of a ground state of the battery and coherence between its energy states, when it comes to amount of extractable average work.

The bounds presented here are a step forward towards understanding to what extent the energy changes on physical batteries can be treated as a genuine thermodynamic work. We emphasize that the deviations from the second law which we described concern the average work rather than the work itself, hence they are of completely different character than deviations reported by standard fluctuation theorems. In the regime of translational symmetry these corrections are negligible, hence the shifts on the battery can be treated as thermodynamic work. This naturally leads to an important open question: what is work if we are not using the battery in the translationally invariant regime? Another interesting investigation in this direction would be to asses the corrections to the second law in terms of the entropy decrease of the system and the bath, rather than in terms of the average work, as we have done in this paper. Furthermore, one can look for saturable upper bounds on the LHS of (\ref{mt:expkb}), so that in the limit $k \rightarrow \infty$ the standard Jarzynski identity can be recovered.

In this work we also laid the first step in answering the fundamental question: how thermodynamic laws and fluctuation relations modify if we take into account the limitations imposed by the presence of the lowest energy state of the battery? This question is relevant both in classical and quantum thermodynamics and is not exclusively related to the resource-theoretic framework which we use here. In this sense we believe that the results we present here open the door for deriving more universal corrections that would apply to arbitrary models of work reservoirs, as well as other frameworks for thermodynamics.

We also showed that the harmonic oscillator battery model correctly reproduces the amount of single-shot work of formation, originally introduced using qubit as the battery system in \cite{Horodecki2013}. Not only this allows to study the two seemingly contradictory notions of \emph{single-shot} and \emph{fluctuating} work using a single battery model, but also, due to Theorem 2, allows to determine when we can treat deterministic work as a genuine thermodynamic work satisfying the second law of thermodynamics. A natural open question is whether this result can be extended to the case of work of distillation, i.e. when work is distilled from a quantum state and stored in the battery.

Furthermore, we demonstrated that in the regime where the population of the battery ground state cannot be neglected, the ideal weight battery does not provide a proper description of physical processes. In the presented example of Landauer erasure we showed that there is a fundamental lower bound on the minimal variance of work distribution that any physical model of a battery has to satisfy. This bound is violated for the ideal weight and shows that fluctuations of thermodynamic work in the low-energy regime behave significantly different than in the regime of high energies.

Finally, we showed that physical batteries will suffer from the existence of the ground state. This influences our understanding of the notion of work as well as asks for reformulation of the whole domain of fluctuation relations whenever the battery is used in the "close to the vacuum" regime. However, one may also ask an orthogonal question of whether physical batteries can offer any kind of advantage over the ideal weight. This must necessary happen in a way that this advantage be "larger" than the violation of the second law as quantified by our correction terms. The arising advantage could not be then attributed only to the second law violation but rather would describe a genuine improvement. This is indeed not excluded since in our model, by accepting a weaker notion of translational invariance, we effectively allow for a larger class of operations than in the case of the ideal weight model. 

\begin{acknowledgments}
We would like to thank Andrzej Grudka, Edgar Aguilar, Mischa Woods, Rafał Demkowicz-Dobrzański, Paul Skrzypczyk, Tony Short, Llouis Masanes, Marcin Łobejko and Jonathan Oppenheim for helpful and inspiring discussions. MH, PLB and PM are supported by National Science Centre, Poland, grant OPUS 9 2015/17/B/ST2/01945. PLB acknowledges support from the UK EPSRC (grant no. EP/R00644X/1). PLB is also supported by National Science Centre, Poland, grant PRELUDIUM 14 2017/27/N/ST2/01227. MH acknowledges support by the Foundation for Polish Science through IRAP project co-financed by EU within Smart Growth Operational Programme (contract no. 2018/MAB/5). This work was also supported by John
Templeton Foundation. 
\end{acknowledgments}

\bibliographystyle{unsrtnat}
\bibliography{citations}

\onecolumn\newpage
\appendix

\section*{Appendices}
In this section we characterize in details the thermodynamic formalism used in the main text (Appendix \ref{appA}) and describe three different battery models used in the literature to define work in the quantum regime (Appendix \ref{appB}). In Appendix \ref{appD} we prove Theorem \ref{theorem1} and in Appendix \ref{app_proof_thm2} we prove Theorem \ref{mt:theorem2}. In Appendix \ref{app_prop_ext_maps} we formally derive the properties of the map $\Gamma_{osc}$.

\section{Thermodynamic framework}
\label{appA}
\subsection{Thermal operations}
The setting of thermal operations consists of a system $\rm S$ with Hamiltonian $H_{\rm S}$ that we would like to apply transformations on, an infinite heat bath $\rm B$ with Hamiltonian $H_{\rm B}$, initially in a Gibbs state $\tau_{\rm B} = \frac{1}{Z_{\rm B}} e^{-\beta H_{\rm B}}$ where $Z_{\rm B}$ is the partition function $Z_{\rm B} = \tr e^{-\beta H_{\rm B}}$, and a battery system $\rm W$ with Hamiltonian $H_{\rm W}$ which we do not define yet. Any joint transformation of the system $\rm S$, bath $\rm B$ and weight $\rm W$ in this framework can be represented by a completely positive trace-preserving (CPTP) channel $\Gamma_{\rm SBW}$ satisfying the following conditions (see \cite{Alhambra2016,Masanes2017,Richens2016,Richens2017} for a more detailed discussion):
\begin{description}
\item[(A) Postulate I]\textit{(strict energy conservation)} 
\begin{align}
\label{app_eq1A}
\left[U, H_{\rm S} + H_{\rm W} + H_{\rm B} \right] = 0.
\end{align} \noindent
This implies that the energy of the joint system $\rm SBW$ is conserved at each time of the action of $\Gamma_{\rm SBW}$. 
\item[(B) Postulate II]\textit{(microscopic reversibility)} \\ \\
\noindent The joint transformation of the system, bath and battery is unitary. Thus, there exists a unitary operator
\begin{align*}
U: \qquad \mathcal{H}_{\rm S} \otimes \mathcal{H}_{\rm W} \otimes \mathcal{H}_{\rm B} \rightarrow \mathcal{H}_{\rm S} \otimes \mathcal{H}_{\rm W} \otimes \mathcal{H}_{\rm B},
\end{align*}
such that:
\begin{align}
\label{app_eq1B}
\Gamma_{\rm SWB}\left[ \rho_{\rm SWB}\right] = U \rho_{\rm SWB} U^{\dagger},
\end{align}
where $U U^{\dagger} = \mathbb{1}_{\rm SWB}$ and $\mathcal{H}_{\rm A}$ denotes the Hilbert space associated with system $\rm A$. In other words, $\Gamma_{\rm SWB}$ has control over all microscopic degrees of freedom of the joint system $\rm SWB$ and no information is dumped into the environment.
\item[(C) Postulate III]\textit{(definition of work)} \\ \\
Before and after applying the global map $\Gamma_{\rm SWB}$ the energy of the battery is measured obtaining outcomes $\ket{\epsilon_k}_{\rm W}$ and $\ket{\epsilon_k+w}_{\rm W}$ respectively. The thermodynamic work $w$ is a random variable with probability distribution given by:
\begin{align}
\label{app_eq1C}
\quad p(w) = \sum_k \tr\left[\Pi_{\epsilon_k}\Gamma_{\rm SWB}\left[ \Pi_{\epsilon_k} \rho_{\rm SW} \Pi_{\epsilon_k}\otimes \tau_{\rm B} \right]\right],
\end{align} 
where $\Pi_{\epsilon_k} = \mathbb{1}_{\rm S} \otimes \dyad{\epsilon_k}_{\rm W} \ot \mathbb{1}_{\rm B}$ are projectors onto the energy eigenstates of the battery $\rm W$. 
\end{description}
\noindent
Thermal operations can be extended to the case when the Hamiltonian changes. This is done by adding an ancillary qubit system $C$ with a trivial Hamiltonian $H_{\rm C} = 0$ called the \emph{switch}. The total Hamiltonian then reads:
\begin{align}
    H = H_{\rm S} \otimes \dyad{0}_{\rm C} + H_{\rm S}' \otimes \dyad{1}_{\rm C} + H_{\rm W} + H_{\rm B}.
\end{align}
Let us now see what the first postulate implies for this model. Let $V$ be the global unitary which acts on $\rm SBW$ and also on the system $C$, i.e.:
\begin{align}
    V = U \ot \dyad{1}{0}_{\rm C} + U' \ot \dyad{0}{1}_{\rm C} 
\end{align}
We also assume that the switch always starts in the state $\dyad{0}_{\rm C}$ and ends up in the state $\dyad{1}_{\rm C}$. The new unitary $V$ then performs the switching as:
\begin{align}
    V(\rho_{\rm SWB} \ot \dyad{0}_{\rm C})V^{\dagger} = \rho_{\rm SWB}' \ot \dyad{1}_{\rm C}, 
\end{align}
for any $\rho_{\rm SWB}$. The first postulate then implies:
\begin{align}
    \label{changing-ham}
    U(H_{\rm S} + H_{\rm W} + H_{\rm B}) = (H_{\rm S} + H_{\rm W} + H_{\rm B}) U 
\end{align}
Effectively this means that the reduced state on the system, bath and the battery is as in (\ref{app_eq1B}), but the unitary $U$ does not necessarily commute with the initial nor final Hamiltonians but satisfies (\ref{changing-ham}).

\noindent
In what follows we will be interested in the joint dynamics of the system and the battery and thus the transformations we consider here are CPTP channels of the form:
\begin{align}
\label{app_eq2}
\Gamma_{\rm SW}\left[\rho_{\rm SW} \right] = \tr_{\rm B} \Gamma_{\rm SWB}\left[\rho_{\rm SW}  \otimes \tau_{\rm B}\right].
\end{align}
We will refer to any $\Gamma_{\rm SW}$ of the above form as a \emph{thermal operation} (TP). It is important to note that the channel $\Gamma_{\rm SW}$ provides only partial information about the action of $\Gamma_{\rm SWB}$, but for our purposes this is enough since we are ultimately interested in the work distribution $p(w)$. In this paper we consider processes where the input and output states of $\rm SW$ are both diagonal in the energy eigenbasis of the joint Hamiltonian $H_{\rm SBW}$. For such states the action of $\Gamma_{\rm SW}$ can be fully encoded in a stochastic matrix $R$ with elements $0 \leq r(s'k'|sk) \leq 1$ defined via: 
\begin{align}
\label{app_eq3}
\Gamma_{\rm SW}  \Big[ \dyad{s}_{\rm S} \otimes  \dyad{\epsilon_k}_{\rm W} \Big] = \sum_{s',k'} r(s'k'|sk) \dyad{s'}_{\rm S} \otimes \dyad{\epsilon_{k'}}_{\rm W}.
\end{align} 
Due to postulates (\textbf{I}-\textbf{III}) matrix elements $r(s'k'|sk)$ must satisfy certain conditions in order to describe a valid thermal operation which takes a diagonal state $\rho_{\rm SW} = \sum_{s,k} p(s,k) \, \dyad{s}_{\rm S} \otimes \dyad{\epsilon_k}_{\rm W}$ to some other diagonal state $\sigma_{\rm SW} = \sum_{s,k} q(s,k)\, \dyad{s}_{\rm S} \otimes \dyad{\epsilon_k}_{\rm W}$. These constraints are: \\
\begin{tcolorbox}[title=Thermal operations]
\small
\begin{align}
\label{app_eq4}
&\forall_{s',k'} \qquad  \qquad  \sum_{s,k} p(s, k)\, r(s'k'|sk) = q(s',k'), \\
\label{app_eq5}
&\forall_{s',k'}  \qquad  \qquad  \sum_{s,k} r(s'k'|sk) e^{-\beta (\epsilon_k - \epsilon_{k'} +E_s - \widetilde{E}_{s'})} = 1 \\
\label{app_eq6}
&\forall_{s,k}  \qquad  \qquad  \, \, \, \sum_{s',k'} r(s'k'|sk) = 1, \\
\label{app_eq7}
&\forall_{s',k',s,k} \qquad  \qquad  \, \, \, r(s'k'|sk) \geq 0,
\end{align}
\normalsize
\end{tcolorbox} \vspace{10pt}
\noindent where $\epsilon_i$ are the energies of the battery $\rm W$ associated with states $\dyad{i}_{\rm W}$ and $E_s$ and $\widetilde{E}_{s'}$ are the energies associated with the system states for the  initial and final Hamiltonian on $S$ respectively. Condition (\ref{app_eq4}) implies that the stochastic map $\Gamma_{\rm SW}$ is able to create state $\sigma_{\rm S}$ from $\rho_{\rm S}$ and condition (\ref{app_eq5}) ensures that the fixed point of the map is the joint Gibbs state $\tau_{\rm S} \otimes \tau_{\rm W} = \left(Z_{\rm S} Z_{\rm W}\right)^{-1} \cdot \sum_{s,k} e^{-\beta (\epsilon_k +E_s)}$. Conditions (\ref{app_eq6}) and (\ref{app_eq7}) ensure that $R$ is a stochastic matrix and thus $\Gamma_{\rm SW}$ is a CPTP channel. It turns out (see \cite{Marshall2011}) that the constraints from (\ref{app_eq4}$-$\ref{app_eq7}) can be satisfied by a suitably chosen stochastic matrix $R$ if and only if the state $\rho_{\rm SW}$ thermomajorizes $\sigma_{\rm SW}$, denoted by $\rho_{\rm SW} \succ_{T} \sigma_{\rm SW}$. Thermomajorization allows to determine which state transformations are possible in terms of thermal operations. The following  definition was originally introduced in \cite{Ruch1976} and presented as thermomajorization in \cite{Horodecki2013}:
\begin{definition}[thermomajorization]
Let $\rho$ and $\sigma$ be two quantum states block-diagonal in the energy eigenbasis and let $\dim{\rho} = \dim{\sigma} = d$. The order relation of thermomajorization, denoted by $\succ_T$, is defined as:
\begin{align}
\rho \succ_T \sigma \iff \, \forall_{k} \quad \sum_{i = 0}^k \frac{p_i}{e^{-\beta E_i}}  \geq \sum_{i = 0}^k \frac{q_i}{e^{-\beta E_i}},
\end{align}
where $p_i$ and $q_i$ are eigenvalues of $\rho$ and $\sigma$, respectively, reordered such that $\frac{p_1}{e^{-\beta E_1}} \geq \frac{p_2}{e^{-\beta E_2}} \ldots \frac{p_d}{e^{-\beta E_d}}$ and $\frac{q_1}{e^{-\beta E_1}} \geq \frac{q_2}{e^{-\beta E_2}} \ldots \frac{q_d}{e^{-\beta E_d}}$, and $E_i$ denotes the $i$-th energy eigenvalue.
\end{definition} \noindent 
Thus, if the spectrum of the state $\sigma_{\rm SW}$ is thermomajorized by the spectrum of $\rho_{\rm SW}$ then there always exists a stochastic matrix $R$ with elements $r(s'k'|sk)$ that satisfy (\ref{app_eq4}$-$\ref{app_eq7}) and it follows, due to \cite{Horodecki2013}, that for diagonal states there always exists a thermal operation with work $\Gamma_{\rm SW}$ that performs this transformation. 

\subsection{The average work and work variance}\label{AppVar}
Let us now study the work required to perform a given process $\Gamma_{\rm SW}$, transforming state $\rho_{\rm SW}$ into another state $\sigma_{\rm SW}$. Consider the input state of the system and battery to be $\rho_{\rm SW} = \sum_{s,k} p_{\rm SW}(s, k) \dyad{s}_{\rm S} \otimes \dyad{\epsilon_k}_{\rm SW}$. The average work $\langle w \rangle$ associated with process $\Gamma_{\rm SW}$ can be calculated explicitly as:
\small
\begin{align}
\label{app_eq8}
\langle w \rangle &= \tr \left[ H_{\rm W} \left(\Gamma_{\rm SW}\left[\rho_{\rm SW}\right] - \rho_{\rm SW}\right)\right] \nonumber \\ \nonumber &= \sum_{s, k} p_{\rm SW}(s,k) \tr \Big[H_{\rm W} \Big(\Gamma_{\rm SW}\big[\dyad{s}_{\rm S} \otimes \dyad{\epsilon_k}_{\rm W}\big] - \dyad{s}_{\rm S} \otimes \dyad{\epsilon_k}_{\rm W} \Big) \Big] \\ \nonumber &\stackrel{(*)}{=} \sum_{s, k} p_{\rm SW}(s,k) \Big[ \sum_{s', k'} r(s'k'|sk) \epsilon_{k'} - \epsilon_k \Big]  \\ 
&\stackrel{(**)}{=}
 \sum_{s', k'} \sum_{s, k} p_{\rm SW}(s, k)\, r(s'k'|sk) w_{kk'}.
\end{align}
\normalsize
where in line $(*)$ we used $H_{\rm W} \ket{\epsilon_k}_{\rm W} = \epsilon_k \ket{\epsilon_k}_{\rm W}$ and in line $(**)$ we labeled the energy difference associated with battery states $\ket{\epsilon_k}_{\rm W}$ and $\ket{\epsilon_{k'}}_{\rm W}$ by $w_{kk'} = \epsilon_{k'} - \epsilon_{k}$. A negative value of this quantity corresponds to a work cost (energy of the battery decreases), while positive to a work gain (energy of the battery increases).

Analogously, for a given work distribution we define its variance as
\small
\begin{align}
\label{app_VAR}
\text{Var}(w) =
 \sum_{s', k'} \sum_w p_{\rm SW}(s, k)\, r(s'k'|sk) (w_{kk'}-\langle w\rangle)^2.
\end{align}
\normalsize

\section{Thermodynamic batteries}
\label{appB}
\subsection{Wit as the battery system}
\noindent
The first model of a battery in the framework of thermal operations was introduced by Horodecki and Oppenheim in \cite{Horodecki2013} who considered the battery to be a two-level system (called by them a \emph{wit}) with Hamiltonian $H_{\rm W} = \delta \cdot\dyad{1}_{\rm W}$. This allowed them to define a notion of deterministic work (see also \cite{Aberg2011}) as the energy difference between the ground state and the excited state of the wit. In this way TO's assisted with a wit battery take the form:
\begin{align}
\Gamma_{wit}\left[ \rho_{\rm S} \otimes \dyad{i}_{\rm W} \right] = \sigma_{\rm S} \otimes \dyad{j}_{\rm W},
\end{align}
where $(i, j) = (0, 1)$ when the transformation stores work in the battery (distillation) or $(i, j) = (1, 0)$ when the transformation consumes work (formation). The deterministic work of transition, denoted here by $w_{\text{det}}$, is defined to be the maximal (distillation) or minimal (formation) value of energy separation $\delta$ for which the input state thermomajorizes the output state, i.e.:
\begin{align}
\label{app_eq11}
\rho_{\rm S} \otimes \dyad{i}_{\rm W} \succ_T \sigma_{\rm S} \otimes \dyad{j}_{\rm W}.
\end{align}
For distillation $w_{\text{det}}$ yields the maximal amount of work that we are guaranteed to extract from the state $\rho$ by converting it into another state $\sigma$. For formation $w_{\text{det}}$ gives the least amount of work that has to be supplied to guarantee the transition from $\rho$ to $\sigma$. This can be further generalized to cases where one allows the transformation to fail with some error probability $\varepsilon$. This is equivalent to transforming the input into a state $\sigma_{\varepsilon}$ which is at most $\varepsilon$-close to the desired output state $\sigma$ according to the trace distance metric, i.e.: $\norm{\sigma-\sigma_{\varepsilon}}_1\leq \varepsilon$.
Let us now denote $\mathcal{B}_{\varepsilon}(\rho) = \{\rho': \, \norm{\rho' - \rho}_1 \leq \varepsilon \}$. When we take the input state to be a thermal state $\tau$ then the associated deterministic work is called \emph{work of formation} and is equal to \cite{Horodecki2013}:
\begin{align}
\label{app_eq13}
w_F^{\varepsilon}(\rho) = \min_{\rho' \in \mathcal{B}^{\varepsilon}(\rho)} \left[ F_{\text{max}}(\rho') - F(\tau) \right],
\end{align}
where $F_{\text{max}}(\rho) = kT  \log \min \left\{ \lambda: \, \rho \leq \lambda \tau \right\}$ and the minimum is taken over all states $\rho'$ that are $\varepsilon$-close to $\rho$. For the case when the output state is a thermal state one defines the \emph{work of distillation} as:
\begin{align}
\label{app_eq14}
w_D^{\varepsilon}(\rho) = \max_{\rho' \in \mathcal{B}(\rho)} \left[ F(\tau) - F_{\text{min}}(\rho') \right]
\end{align}
where $F_{\text{min}}(\rho) =  - k_{\text{B}} T \log \sum_{i} h(\varepsilon,E_i)\, e^{-\beta E_i}$ and $E_i$ is the system's energy with $h(\varepsilon, E_i) \in \{0, 1\}$ is a smoothed binary indicator function, see \cite{Horodecki2013} for details. In the classical limit (obtained by considering infinitely many identical copies of $\rho$ and $\varepsilon$ going to zero) both $w_{F}^{\varepsilon}(\rho)$ and $w_{D}^{\varepsilon}(\rho)$ converge to the difference in standard free energies $F(\rho) = \tr(H\rho) - TS(\rho)$, where $S(\rho)$ is the ordinary von Neuman entropy of state $\rho$.

\subsection{Ideal weight as the battery system}
\label{secA22}
Recent works \cite{Skrzypczyk2014,Richens2016,Alhambra2016} have proposed a different model of a thermodynamic battery: an ideal weight with Hamiltonian $H_{\rm W} = \int_{\mathbb{R}} \text{d} x\, x \dyad{x}_{\rm W}$, where the basis is formed from orthonormal states $\{\ket{x}_{\rm W}, x \in \mathbb{R} \}$ representing position on the weight. Thermal operations acting on system $\rm S$ and the ideal weight are given by a CPTP map defined as in (\ref{app_eq1B}) but with the additional assumption that the global unitary $U$ commutes with translations on the weight as in (\ref{eq:com_rel}). This ensures that the weight cannot be used as an entropy dump, i.e. the entropy of the system $\rm S$ and heat bath $\rm B$ can never decrease by applying this type of transformation. Formally, translational invariance implies that any map $\Gamma_{\rm SBW}$ reduced to the system and bath $\Gamma_{\rm SB}\left[\cdot\right] = \tr_{\rm W} \Gamma_{\rm SBW}\left[(\cdot) \otimes \rho_{\rm W} \right]$ for any initial state of the weight $\rho_{\rm W} = \int_{-\infty}^{\infty} p_{\rm W}(x) \dyad{x}_{\rm W} \, \text{d}\,x$, can be written as a mixture of unitaries $\{u_x\}$ \cite{Masanes2017}:
\begin{align}
\label{app_eq17}
\Gamma_{\rm SB}\left[\cdot\right] = \sum_{x} p_{\rm W}(x) u_x (\cdot) u_x^{\dagger},
\end{align}
where the unitaries depend only on the global unitary $U$ and not on the state $\rho_{\rm W}$. Such a mixture cannot decrease the entropy of the system and the bath (but \emph{can} increase them). In this way the energy difference of the battery may be associated solely with the work exerted by (or extracted from) system $\rm S$.

\subsection{Harmonic oscillator as the battery system}
\label{secA23}
Let us now describe a battery model which will be studied further in the next appendices. The battery's Hamiltonian is taken to be $H_{\rm W} = \sum_{k=0}^{N} \epsilon_k \dyad{\epsilon_k}_{\rm W}$, where $\epsilon_k = k \, \delta$ with a fixed energy gap $\delta$ and $N$ is the battery dimension. The basis is formed from orthonormal states $\{\ket{\epsilon_k}_{\rm W}|\; 0 \leq k \leq N,\, \epsilon_k = k \,\delta \}$ representing the number of fundamental quanta of energy stored in the battery. Thermal operations acting on this battery again have the form from  (\ref{app_eq1B}), however now we put an additional assumption that the effective (i.e. transformation reduced to the system and the battery) transformation arising from the unitary $U$ \emph{almost} commutes with translations on the battery, i.e. commutes with shifts on the battery for all states of the battery above a certain threshold energy $\epsilon_{\text{min}}$.
Formally this means that for all possible transitions, that is for all values of $s, s'$ and all values of $k, k'$ such that $k \geq k_{\text{min}}$ the map $\Gamma_{\rm SW}$ satisfies:
\begin{align}
     \label{st:eti}
     \qquad \Big[\Gamma_{\rm SW}, \text{id}_{\rm S} \ot \Delta_n^{\dagger} [\cdot]\Delta_n\Big] \left[\rho_{\text{SW}}\right] = 0,
\end{align}
for all states of the battery above the threshold energy $\epsilon_{\text{min}}$, i.e. $\forall \, \rho_{\rm SW}$ s.t. $\tr[(\mathbb{1}_{\rm S} \ot \Pi_{\epsilon}) \rho_{SW}] = 0$ for $\epsilon \geq \epsilon_{\text{min}}$. We will refer to (\ref{st:eti}) as the \emph{effective translational invariance} (ETI) property. The following lemma relates this property with transition probabilities $\{r(s'k'|sk)\}$:
\begin{lemma}
The ETI (\ref{st:eti}) implies that the channel probabilites $\{r(s'k'|sk)\}$ for all $s,s',k'$ and all $k$ such that $k_{\text{min}} \leq k \leq N$ satisfy:
\begin{align}
\label{st:eti_cp}
 r(s' k'|s k) = r(s', n+k'|s, n+k),
\end{align} 
for all $n$ such that $0 \leq n+k \leq N$ and $k_{\text{min}} \leq n+k' \leq N$. 
\end{lemma}
\begin{proof}
Consider taking in (\ref{st:eti}) $\rho_{\rm SW} = \dyad{s}_{\rm S} \ot \dyad{\epsilon_k}_{\rm W}$ where $s$ can be any eigenstate of the system $\rm S$ and $k \geq k_{\text{min}}$. We have:
\begin{align}
    \label{sm:lastsimp}
    0 &= \Gamma_{\rm SW}[\dyad{s}_{\rm S} \ot \Delta_n^{\dagger} \dyad{\epsilon_k}_{\rm W}\Delta_n] - \left(\text{id}_{\rm S} \ot \Delta_n^{\dagger} \right)\Gamma_{\rm SW} \left[\dyad{s}_{\rm S} \ot \dyad{\epsilon_k}_{\rm W}\right] \left(\text{id}_{\rm S} \ot \Delta_n\right) \\ \nonumber
    &= \Gamma_{\rm SW}[\dyad{s}_{\rm S} \ot \dyad{\epsilon_{k+n}}_{\rm W}] - \left(\text{id}_{\rm S} \ot \Delta_n^{\dagger} \right)\Gamma_{\rm SW} \left[\dyad{s}_{\rm S} \ot \dyad{\epsilon_k}_{\rm W}\right] \left(\text{id}_{\rm S} \ot \Delta_n\right) \\ \nonumber
    &= \left(\text{id}_{\rm S} \ot \Delta_n \right) \Gamma_{\rm SW}[\dyad{s}_{\rm S} \ot \dyad{\epsilon_{k+n}}_{\rm W}] \left(\text{id}_{\rm S} \ot \Delta_n^{\dagger} \right) - \\ \nonumber
    &\quad \left(\text{id}_{\rm S} \ot \Delta_n\Delta_n^{\dagger} \right)\Gamma_{\rm SW} \left[\dyad{s}_{\rm S} \ot \dyad{\epsilon_k}_{\rm W}\right] \left(\text{id}_{\rm S} \ot \Delta_n\Delta_n^{\dagger}\right),
\end{align}
where in the last two lines we applied the shift operator $\text{id}_{\rm S} \ot \Delta_n$ on both sides of the equation. Consider now the following relation:
\begin{align}
    \label{sm:4}
     \Delta_n\Delta_n^{\dagger} \ket{\epsilon_{k'}}_{\rm W} = \Delta_n \ket{\epsilon_{k'+n}} = \ket{\epsilon_{k'}},
\end{align}
which holds since $k_{\text{min}} \leq n+k' \leq N$. Let us now project the last line in (\ref{sm:lastsimp}) onto $\ket{s'}_{\rm S} \ot \ket{\epsilon_{k'}}_{\rm W}$ for $k_{\text{min}} \leq n+k' \leq N$ and arbitrary $s'$. We have:
\begin{align}
    \left(\bra{s'}_{\rm S} \ot \bra{\epsilon_{k'+n}}_{\rm W}\right) \Gamma_{\rm SW} & \left[\dyad{s}_{\rm S} \ot \dyad{\epsilon_{k+n}}_{\rm W}\right]  \left(\ket{s'}_{\rm S} \ot \ket{\epsilon_{k'+n}}_{\rm W}\right) = \\ \nonumber &\left(\bra{s'}_{\rm S} \ot \bra{\epsilon_{k'}}_{\rm W}\right) \Gamma_{\rm SW}\left[\dyad{s}_{\rm S} \ot \dyad{\epsilon_{k}}_{\rm W}\right] \left(\ket{s'}_{\rm S} \ot \ket{\epsilon_{k'}}_{\rm W}\right),
\end{align}
where we used (\ref{sm:4}). Finally, using the relation between the transition probabilities $\{r(s'k'|sk)\}$ and the map $\Gamma_{\rm SW}$ defined in ($\ref{mt:rcomp}$) we arrive at:
\begin{align}
    r(s'k'|sk) = r(s',k'+n|s,k+n)
\end{align}
which holds for all $s,s'$ and $k,k',n$ satisfying $0 \leq k+n \leq N$ and $k_{\text{min}} \leq k'+n \leq N$.
\end{proof}

\section{Proof of Theorem \ref{theorem1}}
\label{appD}
\noindent  Let us start by defining a random variable: 
\begin{align}
    f_s := E_s + k_{\text{B}} T \log p_{\rm S}(s)
\end{align}
occuring with probability $p_{\rm S}(s)$ and whose average is the free energy of a system in state $\rho_{\rm S} = \sum_s p_{\rm S}(s) \dyad{s}_{\rm S}$, i.e.  
\begin{align}
    \langle f_s \rangle = \sum_{s} p_{\rm S}(s) \, [E_s + k_{\text{B}} T \log p_{\rm S}(s)] = F(\rho_{\rm S})
\end{align}
Now consider the following average:
\begin{align}
    \label{eq:app11}
    \langle e^{\beta(f_{s'} - f_s + w)} \rangle\! &=\! \sum_{s, s', w} p(s, s', w)\, e^{\beta(f_{s'} - f_s + w)}.
\end{align}
We can compute the distribution $p(s, s', w)$ explicitly from the map $\Gamma_{\rm SW}$ (see Postulate III in Appendix \ref{appA}) and write it using transition probabilities $\{r(s'k'|sk)\}$ in the following way:
\begin{align}
\label{eq:probdist}
p(s,s',w) = p_{\rm S}(s)\, p(s'w|s) = \sum_{k,k'} p_{\rm W}(k) \, p_{\rm S}(s)\, r(s'k'|sk)\, \delta_{w,w_{kk'}},
\end{align}
where we defined $w_{kk'} := \epsilon_{k'} - \epsilon_k$. We also define an analogous average quantity conditioned on the battery eigenstate $\ket{\epsilon_k}_{\rm W}$, i.e.: 
\begin{align}
   \langle e^{\beta(f_{s'} - f_s + w_{kk'})} \rangle_k\! :=\! \sum_{s, s', k'} &p_{\rm S}(s)\, r(s'k'|sk) e^{\beta(f_{s'} - f_s + w_{kk'})}. \label{eq:expk}
\end{align}
This allows us to write (\ref{eq:app11}) as:
\begin{align}
     \langle e^{\beta(f_{s'} - f_s + w)} \rangle\! =\! \sum_{k} p_{\rm W}(k)   \langle e^{\beta(f_{s'} - f_s + w_{kk'})} \rangle_k \, \delta_{w,w_{kk'}} \nonumber
\end{align}
We are now ready to prove Theorem \ref{theorem1}.
\begin{proof}
We start by rewritting:
\begin{align}
    \langle e^{\beta(f_{s'} - f_s + w_{kk'})} \rangle_k &= \sum_{s'} p_{\rm S}(s') \, h(s', k),  \label{eq:fluctmain}  
\end{align}
where we labeled $h(s',k) \! =\!  \sum_{s, k'} r(s' k'| s k) e^{\beta(\widetilde{E}_{s'}\! -\! E_s + w_{kk'})}$. It will be also convenient to introduce:
\begin{align}
    \label{gs_def}
    g_{s'}(k'|k) := \sum_s e^{\beta(\widetilde{E}_{s'} - E_s)} \, r(s'k'|sk).
\end{align}
Notice that for each $k \geq k_{\text{min}}$ we can rewrite $h(s', k)$ as:
\begin{align}
\label{eq:maineq} 
h(s', k) &= \sum_{k' = 0}^{N-k_{\text{min}}} g_{s'}(k'|k) e^{\beta(\epsilon_{k'} - \epsilon_k)} + \sum_{k'=N-k_{\text{min}}+1}^N g_{s'}(k'|k)e^{\beta(\epsilon_{k'}-\epsilon_k)} \nonumber \\
&= \sum_{k' = 0}^{N-k_{\text{min}}} g_{s'}(k'|k) e^{\beta(\epsilon_{k'} - \epsilon_k)} + \nonumber \sum_{i = 1}^{k_{\text{min}}} g_{s'}(N - k + i|k_{\text{min}})e^{\beta(\epsilon_{N-k+i}-\epsilon_{\text{min}})}, \nonumber
\end{align}
where in the last equation we changed the summation index to $i = k' - N + k_{\text{min}}$ and used the ETI property (\ref{st:eti_cp}) with $n=k_{\text{min}} - k$. Using again (\ref{st:eti_cp}) and the Gibbs-preserving conditions (\ref{app_eq5}) we can rewrite the first sum in the above equality as:
\begin{align}
\nonumber
\sum_{k' = 0}^{N-k_{\text{min}}} g_{s'}(k'|k)  e^{\beta(\epsilon_{k'} - \epsilon_k)} &= \sum_{k' = 0}^{N-k_{\text{min}}} g_{s'}(N-k|N-k') e^{\beta(\epsilon_{k'} - \epsilon_k)} \nonumber \\
&= \sum_{l' = k_{\text{min}}}^{N} g_{s'}(N-k|l') e^{\beta(\epsilon_{N-l'} - \epsilon_k)}  \nonumber\\
&= 1 - \sum_{l = 0}^{k_{\text{min}}-1} g_{s'}(N-k|l) e^{\beta(\epsilon_{N-k} - \epsilon_l)}  \nonumber\\
&= 1 - A_{s'}(N-k),
\end{align}
where we labelled $A_{s'}(x) := \sum_{l=0}^{k_{\text{min}}-1} g_{s'}(x|l)\, e^{\beta(\epsilon_{x}-\epsilon_l)}$. Using the condition (\ref{app_eq5}) with $k' = N-k$ we get:
\begin{align}
\label{eq:gpc1}
A_{s'}(N-k) + \sum_{l = k_{\text{min}}}^{N-1} & g_{s'}(N-k|l)e^{\beta(\epsilon_{N-k}-\epsilon_l)} + g_{s'}(N-k|N) e^{\beta(\epsilon_{N-k} - \epsilon_N)} = 1.
\end{align}
On the other hand for $k' = N - k + 1$ we get:
\begin{align}
\label{eq:gpc2}
A_{s'}&(N-k+1) +  g_{s'}(N-k+1|k_{\text{min}}) \, e^{\beta(\epsilon_{N-k+1}-\epsilon_{\text{min}})} + \sum_{l = k_{\text{min}} + 1}^{N} g_{s'}(N-k+1|l)e^{\beta(\epsilon_{N-k+1}-\epsilon_l)} = 1.
\end{align} 
Notice that due to the ETI property we also have: 
\begin{align}
\sum_{l = k_{\text{min}} + 1}^{N} & g_{s'}(N-k+1|l)e^{\beta(\epsilon_{N-k+1}-\epsilon_l)} = \sum_{l = k_{\text{min}}}^{N-1} g_{s'}(N-k|l)e^{\beta(\epsilon_{N-k}-\epsilon_l)},
\end{align} 
Combining (\ref{eq:gpc1}) with (\ref{eq:gpc2}) and the last equation yields:
\begin{align}
\nonumber
A_{s'}(N-k) =& A_{s'}(N-k+1) + g_{s'}(N-k+1|k_{\text{min}}) e^{\beta(\epsilon_{N-k+1}-\epsilon_{\text{min}})} - g_{s'}(N-k|N)e^{\beta(\epsilon_{N-k}-\epsilon_N)}.
\end{align}
We now repeat this step $k_{\text{min}}$ times to express $A_{s'}(N-k)$ using $A_{s'}(N-k+k_{\text{min}})$, i.e.
\begin{align}
\label{eq:eq1}
A_{s'}(N-k) =& A_{s'}(N-k+k_{\text{min}}) + \sum_{i=1}^{k_{\text{min}}}  g_{s'}(N-k+i|k_{\text{min}}) e^{\beta(\epsilon_{N-k+i}-\epsilon_{\text{min}})} - \\\nonumber
&\sum_{i=0}^{k_{\text{min}}-1}  g_{s'}(N-k+i|N) e^{\beta(\epsilon_{N-k+i}-\epsilon_N)}.
\end{align}
Notice now that the first sum on the RHS of (\ref{eq:eq1}) appears exactly in (\ref{eq:maineq}). This allows to write for all $k \geq k_{\text{min}}$:
\begin{align}
h(s', k) &= 1 + \sum_{i=0}^{k_{\text{min}}-1}  g_{s'}(N-k+i|N) e^{\beta(\epsilon_{N-k+i}-\epsilon_N)} - A_{s'}(N-k+k_{\text{min}}) \\
&= 1 + \sum_{i=0}^{k_{\text{min}}-1}  g_{s'}(N-k+i|N) e^{\beta(\epsilon_{N-k+i}-\epsilon_N)} - \sum_{i=0}^{k_{\text{min}}-1} g_{s'}(N-k+k_{\text{min}}|i)e^{\beta(\epsilon_{N-k+k_{\text{min}}} - \epsilon_i)} \\
&= 1 + \sum_{i=0}^{k_{\text{min}}-1} \left[ g_{s'}(N-k+i|N) e^{\beta(\epsilon_{i}-\epsilon_k)} - g_{s'}(N-k+k_{\text{min}}|i)e^{\beta(\epsilon_{N-k+k_{\text{min}}} - \epsilon_i)}\right] \\
\label{eq:boundingh}
&\leq 1 + e^{-\beta \epsilon_k} \cdot \sum_{i=0}^{k_{\text{min}}-1} \left[ g_{s'}(N-k+i|N) e^{\beta\epsilon_{i}}\right] \\
&\leq 1 + e^{-\beta \epsilon_k} e^{\beta (\epsilon_{k_{\text{min}}}-\delta)} g_{s'}(N-(k-k_{\text{min}})|N)
\end{align}
Where the first inequality follows since $g_{s'}(k'|k)$ are positive. The last inequality follows since we are looking for a bound which holds for all Gibbs-preserving channels and hence we have to choose the worst-case set of transition probabilities $\{r(s'k'|sk)\}$.  Due to the stochasticity of the channel which implies that for all $s, k$ we have $\sum_{s',k'} r(s'k'|sk) = 1$ w.l.o.g we can choose a probability distribution which maximizes the sum in (\ref{eq:boundingh}). The maximal value of the sum is achieved when the total probability mass is placed in the transition corresponding to output level $i$ with the largest value of $e^{\beta \epsilon_i}$ (i.e. $i = k_{\text{min}}-1$). This also means that we have:
\begin{align}
\label{qnorm}
\forall\, {s,k} \quad \sum_{s',k'} r(s'k'|sk) = 1 \qquad \implies \qquad \forall\, s \quad \sum_{s'} r(s', N-(k-k_{\text{min}})|s, N) = 1.
\end{align}
Let us denote $q(s'|s) := r(s', N-(k-k_{\text{min}})|s, N)$ which due to (\ref{qnorm}) satisfies$\sum_{s'}q(s'|s) = 1$ for each $s$. Using  now (\ref{gs_def}) and defining $\delta_k := \epsilon_k - \epsilon_{\text{min}} + \delta = \delta (k-k_{\text{min}}+1)$ we can further rewrite the bound for $h(s', k)$ as:
\begin{align}
    \forall_{k \geq k_{\text{min}}} \qquad h(s', k) &\leq 1 + e^{-\beta \delta_k} \sum_{s}e^{\beta(\widetilde{E}_{s'} - E_s)} r(s', N-(k-k_{\text{min}})|s, N) \\
    &= 1 + e^{-\beta \delta_k} e^{\beta \widetilde{E}_{s'}} \sum_{s}e^{-\beta E_s} q(s'|s).
\end{align}
Writing above expression explicitly using the definition of $h(s',k)$ from below (\ref{eq:fluctmain}) we get:
\begin{align}
    \label{eq:add1}
    \forall\, {s'},\,\, \forall\, {k \geq k_{\text{min}}}  \qquad \sum_{s, k'} r(s' k'| s k) e^{\beta(\widetilde{E}_{s'}\! -\! E_s + w_{kk'})} \leq 1 + e^{ - \beta \delta_k} e^{\beta \widetilde{E}_{s'}} \sum_{s}e^{-\beta E_s}q(s'|s)
\end{align}
Let us now multiply both sides of (\ref{eq:add1}) by a factor $e^{-\beta \widetilde{E}_{s'}}$ and sum over $s'$. This leads to:
\begin{align}
    \forall\, {k \geq k_{\text{min}}}  \qquad \sum_{s, s', k'} p_{\rm S}(s)\, r(s'k'|sk) e^{\beta(w_{kk'}-f_s)} &\leq Z_{\rm S'} + e^{ - \beta \delta_k} \sum_{s} e^{-\beta E_s}\left(\sum_{s'} q(s'|s)\right) \\
    &= Z_{\rm S'}\left(1+ e^{-\beta\delta_k}\right),
\end{align}
where $Z_{\rm S'}$ is the partition function of the final Hamiltonian on $S$.  Let us now denote the conditional probability distribution $p(s,k'|k) := \sum_{s'} p_{\rm S}(s) \, r(s'k'|sk)$. This allows us to write:
\begin{align}
    \forall\, {k \geq k_{\text{min}}}  \qquad \sum_{s, k'} p(s,k'|k) e^{\beta(w_{kk'}-f_s)} = \langle e^{\beta(w_{kk'} - f_s)} \rangle_k \leq Z_{\rm S'}\left(1 + e^{ - \beta \delta_k} \right),
\end{align}
where averaging $\langle \cdot \rangle_k$ is over random variables $f_{s}$ and $w_{kk'}$ for a given $k$ and can be obtained by setting $f_{s'} = 0$ in (\ref{eq:expk}). This proves the Theorem. Notice further that we prove the inequality by choosing a particular subset of feasible transition probabilities $\{r(s'k'|sk)\}$ and hence the above bound can be saturated.
\end{proof}

It turns out that a more refined version of Theorem \ref{theorem1} will be useful when proving Theorem \ref{mt:theorem2}. Let us return to the expression (\ref{eq:fluctmain}). For all $k \geq k_{\text{min}}$ we have:
\begin{align}
\label{eq:exp_form}
\langle e^{\beta(f_{s'} - f_s + w_{kk'})} \rangle_k= \sum_{s'} p_{\rm S}(s') \, h(s', k)  & \leq 1 +  e^{- \beta \delta_k} \sum_{s,s'}  p_{\rm S}(s') q(s'|s) e^{\beta  (\widetilde{E}_{s'}-E_s)} \\ & \leq 1 +  e^{ - \beta \delta_k} e^{-\beta E_{\rm S'}^{\text{max}}} Z_{\rm S} \\
&= 1 + \eta_{\rm S} e^{ - \beta \delta_k} ,
\end{align}
where $\eta_{\rm S} := Z_{\rm S} e^{\beta E_{\rm S'}^{\text{max}}}$ and $E_{\rm S'}^{\text{max}} := \max_{s'} \widetilde{E}_{s'}$. We will use this expression when proving Theorem \ref{mt:theorem2}.

\section{Proof of Theorem \ref{mt:theorem2}}
\label{app_proof_thm2}
\noindent In this section we present the proof of Theorem \ref{mt:theorem2}.
\begin{proof}
\noindent The average work $\langle w \rangle$ can be written as:
\begin{align}
    \langle w \rangle &= \sum_{w} p(w) \cdot w = \sum_{w} \sum_{k = 0}^N p_{\rm W}(k) \, p(w|k) 
    \label{eq:work_two_terms}
    =\langle w \rangle_{\text{vac}} + \langle w \rangle_{\text{inv}},
\end{align}
where we labeled: 
\begin{align}
\langle w \rangle_{\text{vac}} &:= \sum_w \sum_{k < k_{\text{min}}} p_{\rm W}(k)\, p(w|k) \cdot w, \\
\langle w \rangle_{\text{inv}} &:= \sum_w \sum_{k \geq k_{\text{min}}} p_{\rm W}(k)\, p(w|k) \cdot w
\end{align}
We can further write this using transition probabilities $\{r(s'k'|sk)\}$ as:
\begin{align}
    \langle w \rangle_{\text{vac}} &=\! \sum_{k = 0}^{k_{\text{min}}-1}\! \sum_{s,s',k'} p_{\rm W}(k) \, p_{\rm S}(s) \, r(s'k'|sk)\, w_{kk'},
\end{align}
and similarily for $\langle w \rangle_{\text{inv}}$. The transition probabilities $\{r(s'k'|sk)\}$ satisfy Gibbs-preserving conditions (\ref{app_eq5}), i.e. $r(s'k'|sk) \leq e^{-\beta(\epsilon_{k'} - \epsilon_k)} \, e^{-\beta(\widetilde{E}_{s'} - E_s)}$. This implies:
\begin{align}
    \sum_{s, s'} p_{\rm S}(s) \, r(s'k'|sk) &\leq  \sum_{s, s'} p_{\rm S}(s) \, e^{-\beta w_{kk'}} \, e^{-\beta(\widetilde{E}_{s'} - E_s)} = e^{\beta E_{\rm S'}^{\text{max}}} e^{-\beta w_{kk'}} Z_{\rm S} = \eta_{\rm S} e^{-\beta w_{kk'}},
\end{align}
Plugging this into our expression for $\langle w \rangle_{\text{vac}}$ gives:
\begin{align}
    \langle w \rangle_{\text{vac}} &\leq \eta_{\rm S} \cdot \sum_{k = 0}^{k_{\text{min}}-1} \sum_{k' = 0}^N p_{\rm W}(k) \, e^{-\beta(\epsilon_{k'} - \epsilon_k)} \cdot (\epsilon_{k'} - \epsilon_{k})  \nonumber \\ 
    \nonumber
    &= - \eta_{\rm S} \cdot \sum_{k = 0}^{k_{\text{min}}-1}  p_{\rm W}(k) \cdot \frac{\partial}{\partial \beta} \, \left( \sum_{k' = 0}^N e^{-\beta(\epsilon_{k'} - \epsilon_k)}\right) \\
    &= - \eta_{\rm S} \, \frac{\partial}{\partial \beta} \left[ \sum_{k = 0}^{k_{\text{min}}-1} \frac{p_{\rm W}(k)}{g_{\rm W}(k)}\right] \label{eq:w_vac} 
\end{align}
Consider now the second term from (\ref{eq:work_two_terms}):
\begin{align}
   \langle w \rangle_{\text{inv}} = \sum_{k = k_{\text{min}}}^{N}\sum_{s,s',k'} p_{\rm W}(k) \, p_{\rm S}(s) \, r(s'k'|sk) \, w_{kk'}. \nonumber
\end{align}
 Notice that for $k \geq k_{\text{min}}$ we satisfy the assumptions of Theorem \ref{theorem1}. Thus, multiplying (\ref{eq:exp_form}) by $p_{\rm W}(k)$ and summing over $k \geq k_{\text{min}}$ we obtain the following bound:
\begin{align}
    \label{eq:eq:2}
    \sum_{k = k_{\text{min}}}^N \sum_{s,s',k'} p_{\rm W}(k) p_{\rm S}(s) r(s'k'|sk) e^{\beta (f_{s'} - f_s + w_{kk'}) } \leq  1 + \eta_{\rm S} \Bigg(\sum_{k=k_{\text{min}}}^N p_{\rm W}(k) e^{-\beta \delta_k} \Bigg).
\end{align}
Using the convexity of the exponential function, taking the logarithm of both sides of (\ref{eq:eq:2}) and using the definition of $B_{\beta}(\rho)$ from the theorem we find:
\begin{align}
    \label{eq:eq2}
    B_{\beta}(\rho_{\rm W}) &\geq \!\! \sum_{k= k_{\text{min}}}^N \sum_{\substack{s,s',k'}}\!\!  p_{\rm S}(s) p_{\rm W}(k) r(s'k'|sk) (f_{s'}\!-\!f_s\!+\!w_{kk'}) \nonumber \\ \nonumber
    &= \sum_{k= k_{\text{min}}}^N \! \sum_{\substack{s,s',k'}}\!\!  p_{\rm S}(s) p_{\rm W}(k) r(s'k'|sk) (f_{s'}\!-\!f_s) + \langle w \rangle_{\text{inv}} \\
    &= \Delta F_S - \sum_{k=0}^{k_{\text{min}}-1}\sum_{s,s',k'} p_{\rm S}(s)p_{\rm W}(k) r(s'k'|sk) (f_{s'}-f_s) \nonumber + \langle w \rangle_{\text{inv}}.
\end{align}
where we used:
\begin{align}
    \Delta F_{\rm S} = \sum_{s,k,s',k'} p_{\rm S}(s)p_{\rm W}(k) r(s'k'|sk) (f_{s'}-f_s).
\end{align}
Notice now that we have the following bound:
\begin{align}
    \sum_{k= 0}^{k_{\text{min}} - 1} \sum_{\substack{s,s',k'}}\!\!  p_{\rm S}(s) p_{\rm W}(k) r(s'k'|sk)  f_{s'} &\leq \sum_{k= 0}^{k_{\text{min}} - 1} \sum_{\substack{s'}}\! p_{\rm W}(k) r(s'|k) f_{s'}\\ &\leq \left(\sum_{k= 0}^{k_{\text{min}} - 1}  p_{\rm W}(k)\right) \max_{s'} f_{s'} \\ &= \left(\sum_{k= 0}^{k_{\text{min}} - 1}  p_{\rm W}(k)\right) E_{\rm S'}^{\text{max}}.\label{eq:eq3}
\end{align}
On the other hand we also have:
\begin{align}
    \label{eq:eq4}
    \sum_{k= 0}^{k_{\text{min}} - 1}  \sum_{\substack{s,s',k'}}\!\! p_{\rm S}(s) p_{\rm W}(k)& r(s'k'|sk)  f_{s} = F(\rho_{\rm S}) \cdot   \left(\sum_{k= 0}^{k_{\text{min}} - 1}  p_{\rm W}(k)\right),
\end{align}
where we used the fact that $\sum_{s',k'} r(s'k'|sk) = 1$ for all $s, k$ and $F(\rho_{\rm S}) = \sum_s p_{\rm S}(s)f_s$. Combining (\ref{eq:eq3}) and (\ref{eq:eq4}) gives:
\begin{align}
    \label{eq:w_inv}
    \langle w \rangle_{\text{inv}} \leq & -\Delta F_{\rm S} +B_{\beta}(\rho_{\rm W}) +  \left(\sum_{k= 0}^{k_{\text{min}} - 1}  p_{\rm W}(k)\right) \cdot (E_{\rm S'}^{\text{max}} - F(\rho_{\rm S})). 
\end{align}
The theorem is proven by combining bounds (\ref{eq:w_vac}) and (\ref{eq:w_inv}). 
\end{proof}

\emph{Proof of Corrolary \ref{cor1}}. 
Let us denote the populations by: $x_{<}=\sum_{k<k_{\text{min}}} p_{\rm W}(k)$,
$x_{>}=\sum_{k_{\text{min}}\leq k< k^{*}} p_{\rm W}(k)$, and $1-x=\sum_{k\geq k^{*}} p_{\rm W}(k)$, such that $x_{>}+x_{<}=x$ is the occupation of the battery below the energy cut-off $\epsilon_* = k_{*}\delta$. By rewriting (\ref{eq:w_vac}), we obtain 
\begin{align} 
\langle w \rangle_{\text{vac}} &\leq \eta_{\rm S} \sum_{k = 0}^{k_{\text{min}}-1} \sum_{k' = 0}^N p_{\rm W}(k) \, e^{-\beta(\epsilon_{k'} - \epsilon_k)} \cdot (\epsilon_{k'} - \epsilon_{k}) \label{eq:over1a}
= \eta_{\rm S} Z_{\rm W} \left( \;\,\,\sum_{\mathclap{k = 0}}^{\mathclap{k_{\text{min}}-1}} p_{\rm W}(k) e^{\beta\epsilon_{k}}(\langle E\rangle_{\beta}-\epsilon_{k})  \right),    
\end{align}
with $Z_{\rm W}=\sum_{k'=0}^{N}e^{-\beta \epsilon_{k'}}$ and $\langle E\rangle_{\beta}= {Z^{-1}_{\rm W}} \sum_{k'=0}^{N} e^{-\beta \epsilon_{k'}}\epsilon_{k'}$. We see that (\ref{eq:over1a}) is upper bounded by 
\begin{align}
\langle w \rangle_{\text{vac}} & \leq x_{<} \cdot \eta_{\rm S}\, Z_{\rm W}\,  e^{\beta (k_{\text{min}}-1) x} \langle E\rangle_{\beta} \leq \eta_{\rm S}\, Z_{\rm W} \, x \, e^{\beta \epsilon_{\text{min}}}\langle E\rangle_{\beta}.
\end{align}
On the other hand notice that the free energy difference $\langle f_{s'} - f_s \rangle_{\text{vac}}$ can be upper bounded as:
\begin{align}
    \langle f_{s'} - f_s \rangle_{\text{vac}} &\leq \left(\sum_{k=0}^{k_{\text{min}}} p_{\rm W}(k) \right) \cdot \left( E_{\rm S'}^{\text{max}} - F_{\rm S} \right) \leq x \left(E_{\rm S'}^{\text{max}}+\frac{1}{\beta}\log d_{\rm S}\right).
\end{align}
Combining the last two expressions yields an upper-bound for $A_{\beta}(\rho_W,\rho_{\rm S})$:
\begin{align} \label{eq:over1}
    A_{\beta}(\rho_{\rm W},\rho_{\rm S}) \leq x \Big(\eta_{\rm S}\, Z_{\rm W}  e^{\beta \epsilon_{\text{min}}}\langle E\rangle_{\beta} +  E_{\rm S'}^{\text{max}} + \frac{1}{\beta} \log d_{\rm S} \Big)
\end{align}
Furthermore, the correction stemming from the distance to the vacuum regime can be rewritten as 
\begin{align}
\beta B_{\beta}(\rho_{\rm W}) & \leq \log \Big[ 1+\eta_{\rm S}\sum_{k=k_{\text{min}}}^{k^{*}-1}p_{\rm W}(k) e^{-\beta\delta(k-k_{\text{min}}+1)}+  \eta_{\rm S}  \sum_{k=k_{*}}^{N}p_{\rm W}(k) e^{-\beta\delta(k-k_{\text{min}}+1)}\Big)\Big] \\ &\leq \eta_{\rm S} \, \Big( x_{>}\, e^{\beta\Delta}+(1-x)e^{-\beta\delta(k^{*}-k_{\text{min}}+1)}  \Big) \\ \label{eq:over3}
&\leq \eta_{\rm S}\big(x e^{-\beta\delta}+e^{-\beta D}\big),
\end{align}
with $D:=\delta(k^{*}-k_{\text{min}})$. In the limit $N\rightarrow\infty$, we have $Z_{\rm W}=(1-e^{-\beta\delta})^{-1}$ and $\langle E\rangle_{\beta}= \delta \cdot e^{-\beta\delta} (1-e^{-\beta\delta})^{-1}$. Therefore, combining bounds (\ref{eq:over1}) with (\ref{eq:over3}), together with $\eta_{\rm S} \leq d_{\rm S} e^{\beta E_{\rm S'}^{\text{max}}}$, allows us to find the following bound:
\begin{equation}
 \langle w \rangle \leq - \Delta F_{\rm S} + C(\epsilon_*),
 \end{equation}
 with:
 \begin{align}
     C(\epsilon_{*}) &= x(\epsilon_{*}) \left( c_{\rm S} e^{-\beta\delta}\left(\frac{1}{\beta}+\frac{\delta e^{\beta \epsilon_{\text{min}}}}{(1-e^{-\beta\delta})^{2} }\right) +\frac{1}{\beta} \log c_{\rm S} \right) + \frac{c_{\rm S}}{\beta} e^{-\beta D(\epsilon_*)}, 
 \end{align} 
 where $c_{\rm S} :=  d_{\rm S}e^{\beta E_{\rm S'}^{\text{max}}}$. We also explicitly marked the dependence of $D(\epsilon_{*})$ and $\delta(\epsilon_{*})$ on a selection of threshold energy $\epsilon_{*}$ for a given initial state of the battery. 
 
\section{Properties of the map
$\Gamma_{osc}$}
\label{app_prop_ext_maps}
\subsubsection{The extended map is a thermal operation}
In this section we show that the map defined via Construction \ref{con1} is a valid thermal operation. This is equivalent to showing that $\Gamma_{osc}$ preserves the Gibbs state, i.e. $\Gamma_{osc}\left[\tau_{\rm S} \ot \tau_{\rm W}\right] = \tau_{\rm S} \ot \tau_{\rm W}$. Let us begin by taking an arbitrary thermal operation $\Gamma_{wit}$ acting on the wit. Its action in terms of subchannels $\{\mathcal{R}_{kk'}\}$ can be written as:
\begin{align}
    \Gamma_{wit}[\rho_{\rm S} \ot \rho_{\rm W} ] &= \sum_{k,k'}^{N=1} p_{\rm W}(k) \, \mathcal{R}_{kk'}(\rho_{\rm S}) \ot \dyad{\epsilon_{k'}}_{\rm W}.
\end{align}
Since by the assumption $\Gamma_{wit}$ preserves the Gibbs state we can write:
\begin{align}
    \label{eq:x1}
    \mathcal{R}_{00}(\tau_{\rm S}) + e^{-\beta \delta} \mathcal{R}_{10}(\tau_{\rm S}) &= \tau_{\rm S}, \\
    \label{eq:x2}
    \mathcal{R}_{01}(\tau_{\rm S}) + e^{-\beta \delta} \mathcal{R}_{11}(\tau_{\rm S}) &= e^{-\beta \delta} \tau_{\rm S}.
\end{align}
To show that the map $\Gamma_{osc}$ is a valid thermal operation we will first consider a finite-dimensional harmonic oscillator battery (i.e. we additionally specify the action of the map on the highest energy level $N$) and show that this completed map, denoted $\widetilde{\Gamma}_{osc}$, is a thermal operation for all $N$. We will then take the limit $N \rightarrow \infty$ and show that in this limit the completed map is indistinguishable from the map described via Construction \ref{cor1} in the main text. 

Let us consider a finite dimensional harmonic oscillator battery with Hamiltonian:
\begin{align}
    \label{app_fham}
    H_{\rm W} = \sum_{k=0}^N \epsilon_k \dyad{\epsilon_k}_{\rm W},
\end{align}
where as before $\epsilon_k = k \delta$.
\newline \newline 
\noindent (a) \emph{Gibbs-preserving property.}
We define the completed map $\widetilde{\Gamma}_{osc}$ as:
\begin{align}
    &\rho_{\rm S} \ot \dyad{\epsilon_0}_{\rm W} \rightarrow \sum_{k' = 0}^{N-1} \mathcal{R}_{00} \mathcal{R}_{01}^{k'}(\rho_{\rm S}) \ot \dyad{\epsilon_{k'}}_{\rm W} + \mathcal{R}_{01}^{N}(\rho_{\rm S}) \ot \dyad{\epsilon_N}_{\rm W},  \\  
    &\rho_{\rm S} \ot \dyad{\epsilon_k}_{\rm W} \rightarrow \mathcal{R}_{10}(\rho_{\rm S}) \ot \dyad{\epsilon_{k-1}}_{\rm W} + \sum_{k' = k}^{N-1} \mathcal{R}_{00} \mathcal{R}_{01}^{k'-k} \mathcal{R}_{11}(\rho_{\rm S}) \ot \dyad{\epsilon_{k'}}_{\rm W} + \\ &\hspace{85pt}\mathcal{R}_{01}^{N-k} \mathcal{R}_{11}(\rho_{\rm S}) \ot \dyad{\epsilon_{N}}_{\rm W}, \\
    &\rho_{\rm S} \ot \dyad{\epsilon_N}_{\rm W} \rightarrow \mathcal{R}_{10}(\rho_{\rm S}) \ot \dyad{\epsilon_{N-1}}_{\rm W} + \mathcal{R}_{11}(\rho_{\rm S}) \ot \dyad{\epsilon_{N}}_{\rm W}
\end{align}
where $0 < k < N$. Let us now apply the extended map $\widetilde{\Gamma}_{osc}$ to the Gibbs state $\tau_{\rm S} \ot \tau_{\rm W}$ and look at a single energy level $\dyad{\epsilon_{k'}}_{\rm W}$ of the battery, i.e:
\newline \newline
\noindent For $k' = 0$ we have:
\begin{align}
    \tr_{\rm W} \Big[\left(\mathbb{1}_{\rm S} \ot \dyad{\epsilon_{0}}_{\rm W}\right)  \widetilde{\Gamma}_{osc}(\tau_{\rm S} \ot \tau_{\rm W}) \left(\mathbb{1}_{\rm S} \ot \dyad{\epsilon_{0}}_{\rm W}\right) \Big] &= g_{\rm W}(0) \, \mathcal{R}_{00} + g_{\rm W}(1) \mathcal{R}_{10}(\tau_{\rm S}) = g_{\rm W}(0) \tau_{\rm S}.
\end{align}
\newline
\noindent For $0 < k' < N$ we have:
\begin{align}
    &\tr_{\rm W} \Big[\left(\mathbb{1}_{\rm S}  \ot \dyad{\epsilon_{k'}}_{\rm W}\right) \widetilde{\Gamma}_{osc}(\tau_{\rm S} \ot \tau_{\rm W}) \left(\mathbb{1}_{\rm S} \ot \dyad{\epsilon_{k'}}_{\rm W}\right) \Big] = \\
    &= g_{\rm W}(0) \, \mathcal{R}_{00} \mathcal{R}_{01}^{k'}(\tau_{\rm S}) + \sum_{k=1}^{k'} g_{\rm W}(k)\, \mathcal{R}_{00} \mathcal{R}_{01}^{k'-k} \mathcal{R}_{11}(\tau_{\rm S}) + g_{\rm W}(k'+1) \mathcal{R}_{10}(\tau_{\rm S})  \\ 
    &= g_{\rm W}(0) \mathcal{R}_{00}\mathcal{R}_{01}^{k'-1}\left[\mathcal{R}_{01}(\tau_{\rm S}) + e^{-\beta \delta} \mathcal{R}_{11}(\tau_{\rm S})\right] + \sum_{k=2}^{k'} g_{\rm W}(k)\, \mathcal{R}_{00} \mathcal{R}_{01}^{k'-k} \mathcal{R}_{11}(\tau_{\rm S}) + g_{\rm W}(k'+1) \mathcal{R}_{10}(\tau_{\rm S}) \\
    &= \quad \ldots \text{repeat $(k'-1)$ times} \ldots  \\
    &= g_{\rm W}(k'-1) \mathcal{R}_{00} \left[\mathcal{R}_{01}(\tau_{\rm S}) + e^{-\beta \delta} \mathcal{R}_{11}(\tau_{\rm S})\right] + \sum_{k=k'}^{k'} g_{\rm W}(k)\, \mathcal{R}_{00} \mathcal{R}_{01}^{k'-k} \mathcal{R}_{11}(\tau_{\rm S}) + g_{\rm W}(k'+1)  \mathcal{R}_{10}(\tau_{\rm S}) \\ 
    &= g_{\rm W}(k') \mathcal{R}_{00} (\tau_{\rm S}) + g_{\rm W}(k'+1)  \mathcal{R}_{10}(\tau_{\rm S})  \\ 
    &=  g_{\rm W}(k') \left[\mathcal{R}_{00} (\tau_{\rm S}) + e^{-\beta \Delta}  \mathcal{R}_{10}(\tau_{\rm S})\right] \\
    &= g_{\rm W}(k') \tau_{\rm S}.
\end{align}
\newline
\noindent For $k' = N$ we have:
\begin{align}
    \tr_{\rm W} \Big[\left(\mathbb{1}_{\rm S} \ot \dyad{\epsilon_{N}}_{\rm W}\right) & \widetilde{\Gamma}_{osc}(\tau_{\rm S} \ot \tau_{\rm W}) \left(\mathbb{1}_{\rm S} \ot \dyad{\epsilon_{N}}_{\rm W}\right) \Big] \\
    &= g_{\rm W}(0) \mathcal{R}_{01}^N(\tau_{\rm S}) + \sum_{k=1}^N g_{\rm W}(k) \mathcal{R}_{01}^{N-k}\mathcal{R}_{11}(\tau_{\rm S}) \\ 
    &= g_{\rm W}(0) \mathcal{R}_{01}^{N-1} \left[\mathcal{R}_{01}(\tau_{\rm S}) + e^{-\beta \delta} \mathcal{R}_{11}(\tau_{\rm S}) \right] + \sum_{k=2}^N g_{\rm W}(k) \mathcal{R}_{01}^{N-k}\mathcal{R}_{11}(\tau_{\rm S})
    \\
    &= g_{\rm W}(1) \mathcal{R}_{01}^{N-1}(\tau_{\rm S}) + \sum_{k=2}^N g_{\rm W}(k) \mathcal{R}_{01}^{N-k}\mathcal{R}_{11}(\tau_{\rm S})
    \\     &= \quad \ldots \text{repeat $(N-1)$ times} \ldots \nonumber \\
    &= g_{\rm W}(N-1) \mathcal{R}_{01} (\tau_{\rm S}) + g_{\rm W}(N) \mathcal{R}_{11}(\tau_{\rm S}) \\
    &= g_{\rm W}(N) \tau_{\rm S}. \nonumber
\end{align}

\noindent The total state of the system and the battery is then given by:
\begin{align}
    \widetilde{\Gamma}_{osc}[\tau_{\rm S} \ot \tau_{\rm W}] = \sum_{k'=0}^N g_{\rm W}(k') \tau_{\rm S} \ot \tau_{\rm W} = \tau_{\rm S} \ot \tau_{\rm W}.
\end{align}
\newline \newline 
\noindent (b) \emph{Trace-preserving property}
It can be also checked that the map $\Gamma_{osc}$ is trace preserving since $\Gamma_{wit}$ is trace preserving, i.e. for any $\rho_{\rm S}$ we have:
\begin{align}
    \tr \mathcal{R}_{00}[\rho_{\rm S}] + \tr \mathcal{R}_{01}[\rho_{\rm S}] = \tr \rho_{\rm S}, \\
    \tr \mathcal{R}_{10}[\rho_{\rm S}] +  \tr \mathcal{R}_{11}[\rho_{\rm S}] = \tr \rho_{\rm S}.
\end{align}
\noindent Hence for $k = 0$ we have:
\begin{align}
    \tr \widetilde{\Gamma}_{osc}[\rho_{\rm S} \ot \dyad{\epsilon_{0}}] &= \sum_{k'=0}^{N-1} \tr \mathcal{R}_{00}\mathcal{R}_{01}^{k'}(\rho_{\rm S}) + \tr \mathcal{R}_{01}^N(\rho_{\rm S}) \\
    &=  \tr \mathcal{R}_{00}(\rho_{\rm S}) + \tr \mathcal{R}_{00}\mathcal{R}_{01}(\rho_{\rm S}) + \ldots + \tr \mathcal{R}_{00}\mathcal{R}_{01}^{N-1}(\rho_{\rm S}) + \tr \mathcal{R}_{01}^N(\rho_{\rm S}) \\
    &=  \tr \mathcal{R}_{00}(\rho_{\rm S}) + \tr \mathcal{R}_{00}\mathcal{R}_{01}(\rho_{\rm S}) + \ldots + \tr \mathcal{R}_{01}^{N-2}(\rho_{\rm S}) \\
    &= \quad \ldots \text{repeat $(N-1)$ times} \ldots \\
    &=  \tr \mathcal{R}_{00}(\rho_{\rm S}) + \tr \mathcal{R}_{01}(\rho_{\rm S}) \\ 
    &= \tr \rho_{\rm S}.
\end{align}
\noindent For $0 < k < N$ we have:
\begin{align}
    \tr \widetilde{\Gamma}_{osc}&[\rho_{\rm S} \ot \dyad{\epsilon_{k}}] = \tr \mathcal{R}_{10}(\rho_{\rm S}) + \sum_{k'=k}^{N-1} \tr \mathcal{R}_{00}\mathcal{R}_{01}^{k'-k}\mathcal{R}_{11}(\rho_{\rm S}) + \tr \mathcal{R}_{01}^{N-k}\mathcal{R}_{11}(\rho_{\rm S}) \\
    &=  \tr \mathcal{R}_{10}(\rho_{\rm S}) + \tr \mathcal{R}_{00}\mathcal{R}_{11}(\rho_{\rm S}) + \mathcal{R}_{00}\mathcal{R}_{01}\mathcal{R}_{11}(\rho_{\rm S}) + \ldots + \tr \mathcal{R}_{00} \mathcal{R}_{01}^{N-k-1}\mathcal{R}_{11}(\rho_{\rm S})  \tr \mathcal{R}_{01}^{N-k}\mathcal{R}_{11}(\rho_{\rm S}) \\
    &=  \tr \mathcal{R}_{10}(\rho_{\rm S}) + \tr \mathcal{R}_{00}\mathcal{R}_{11}(\rho_{\rm S}) + \mathcal{R}_{00}\mathcal{R}_{01}\mathcal{R}_{11}(\rho_{\rm S}) + \ldots + \tr \mathcal{R}_{01}^{N-k-1}\mathcal{R}_{11}(\rho_{\rm S})\\
    &= \quad \ldots \text{repeat $(N-k-1)$ times} \ldots \\
    &=  \tr \mathcal{R}_{10}(\rho_{\rm S}) + \tr \mathcal{R}_{11}(\rho_{\rm S}) \\ 
    &= \tr \rho_{\rm S}.
\end{align}
\noindent For $k = N$ we have:
\begin{align}
    \tr \widetilde{\Gamma}_{osc}[\rho_{\rm S} \ot \dyad{\epsilon_{N}}] &= \tr \mathcal{R}_{10}(\rho_{\rm S}) + \mathcal{R}_{11}(\rho_{\rm S}) = \tr \rho_{\rm S}.
\end{align}
Hence $\widetilde{\Gamma}_{osc}$ is trace-preserving which completes the proof that it is a valid thermal operation. 

Let us now take the limit $N\rightarrow \infty$ and study the action of the map $\widetilde{\Gamma}_{osc}$. Notice that due to (\ref{eq:x2}) for all positive $\delta$ we have that $\mathcal{R}_{01}^N(\rho_{\rm S}) \rightarrow 0$ for $N \rightarrow \infty$ for all $\rho_{\rm S}$. Hence in the limit the map $\widetilde{\Gamma}_{osc}$ becomes:
\begin{align}
    &\rho_{\rm S} \ot \dyad{\epsilon_0}_{\rm W} \rightarrow \sum_{i = 0}^{\infty} \mathcal{R}_{00} \mathcal{R}_{01}^{i}(\rho_{\rm S}) \ot \dyad{\epsilon_{i}}_{\rm W},  \\  
    &\rho_{\rm S} \ot \dyad{\epsilon_k}_{\rm W} \rightarrow \mathcal{R}_{10}(\rho_{\rm S}) \ot \dyad{\epsilon_{k-1}}_{\rm W} + \sum_{i = 0}^{\infty} \mathcal{R}_{00} \mathcal{R}_{01}^{i} \mathcal{R}_{11}(\rho_{\rm S}) \ot \dyad{\epsilon_{i+k}}_{\rm W},
\end{align}
where $0 < k < \infty$, and hence it is indistinguishable from $\Gamma_{osc}$.

\subsubsection{The average work under $\Gamma_{osc}$}
\noindent It is easy to see that the action of $\Gamma_{osc}$ does not depend on the initial state of the battery $\rho_{\rm W}$, provided that the state does not have support on the vacuum state $\dyad{\epsilon_0}_{\rm W}$. This allows to easily compute the average work $\langle w \rangle_N$ associated with $\Gamma_{osc}$ as well as its action on the system $S$.

Let us begin by first studying the corresponding map $\widetilde{\Gamma}^f_{osc}$ acting on a finite $N+1$ dimensional harmonic oscillator battery. The energy change of the battery after applying $\widetilde{\Gamma}_{osc}$ to an input state $\rho_{\rm SW} = \rho_{\rm S} \ot \dyad{\epsilon_k}_{\rm}$ will in general be different for different $k$. The energy change $\Delta E_W^{(k)}$ associated with battery starting in eigenstate $\dyad{\epsilon_k}_{\rm W}$ is then given by:
\begin{align}
\Delta E_{\rm W}^{(k)} = \tilde{E}_{\rm W}^{(k)} - E_{\rm W}^{(k)},
\end{align}
where $E_{\rm W}^{(k)} = \langle \epsilon_k|H_{\rm W}|\epsilon_k \rangle = \epsilon_k$, $\tilde{E}_{\rm W}^{(k)} = \tr \left[H_{\rm W} \rho'_{\rm W}\right]$, $H_{\rm W}$ is the finite harmonic oscillator Hamiltonian as in (\ref{app_fham}) and $\rho_{\rm W}' = \tr_{\rm S} \widetilde{\Gamma}^f_{osc} \left[\rho_{SW}\right]$. Denoting the number of energy levels above $\epsilon_k$ with $n$, i.e $n = N-k$, we can compute the final energy $\tilde{E}_{\rm W}^{(k)}$ as: 

\begin{align}
\tilde{E}_{\rm W}^{(k)} &= (\epsilon_k - \delta) \tr \mathcal{R}_{10} \left[\rho\right] + \sum_{j=0}^{n-1} (\epsilon_k + j \cdot \Delta) \, \tr \mathcal{R}_{00} \mathcal{R}_{01}^{j} \mathcal{R}_{11} \left[\rho_{\rm S}\right] + (\epsilon_k + n \cdot \delta)\, \tr \mathcal{R}_{01}^{n} \mathcal{R}_{11} \left[\rho_{\rm S} \right] \\ 
&= (\epsilon_k - \delta) \cdot \big( \tr \mathcal{R}_{10} \left[\rho_{\rm S}\right] +  \sum_{j=0}^{n-1} \tr \mathcal{R}_{00} \mathcal{R}_{01}^{j} \mathcal{R}_{11} \left[\rho_{\rm S}\right] + \tr \mathcal{R}_{10}^n  \mathcal{R}_{11} \left[\rho_{\rm S}\right] \big) + \\
&\qquad \delta \cdot \sum_{j=0}^{n-1} \left(j+1 \right) \, \tr \mathcal{R}_{00} \mathcal{R}_{01}^{j} \mathcal{R}_{11} \left[\rho_{\rm S}\right] + \delta\cdot (n+1) \cdot \, \tr \mathcal{R}_{01}^{n} \mathcal{R}_{11} \left[\rho_{\rm S} \right]  \\ 
&= (\epsilon_k - \delta) \cdot \left( \tr\mathcal{R}_{10} \left[\rho_{\rm S}\right] + \tr \mathcal{R}_{11} \left[\rho_{\rm S}\right] \right) + \delta \cdot  \sum_{j=0}^{n-2} (j+1) \tr \mathcal{R}_{00} \mathcal{R}_{01}^{j} \mathcal{R}_{11} [\rho_{\rm S}]  + \\ 
&\qquad n \cdot \tr \mathcal{R}_{00} \mathcal{R}_{01}^{n-1} \mathcal{R}_{11} [\rho_{\rm S}] +  n \cdot \tr \mathcal{R}_{01} \mathcal{R}_{01}^{n-1} \mathcal{R}_{11} [\rho_{\rm S}] + \tr \mathcal{R}_{01}^n \mathcal{R}_{11} [\rho_{\rm S}] \\
&\stackrel{(2)}{=} (\epsilon_k - \delta) + \delta \cdot \sum_{j=0}^{n-2} (j+1) \tr \mathcal{R}_{00} \mathcal{R}_{01}^{j} \mathcal{R}_{11} [\rho_{\rm S}] + \delta \cdot n \cdot \tr \mathcal{R}_{01}^{n-1} \mathcal{R}_{11} [\rho_{\rm S}] + \tr \mathcal{R}_{01}^n  \mathcal{R}_{11} [\rho_{\rm S}] \\
& = (\epsilon_k - \delta) + \tr \mathcal{R}_{11}[\rho_{\rm S}] +  \tr \mathcal{R}_{01}  \mathcal{R}_{11}[\rho_{\rm S}]  + \tr \mathcal{R}_{01}^2  \mathcal{R}_{11}[\rho_{\rm S}] + \ldots + \tr \mathcal{R}_{01}^n [\rho_{\rm S}]  \mathcal{R}_{11}[\rho_{\rm S}] \\ 
&= (\epsilon_k - \delta) + \sum_{j=0}^n \tr \mathcal{R}_{01}^j \mathcal{R}_{11}[\rho_{\rm S}] \\ 
&= \epsilon_k + \delta\cdot \left(\sum_{j=0}^n \tr \mathcal{R}_{01}^j \mathcal{R}_{11}[\rho_{\rm S}] - 1\right). \numberthis
\end{align}
In line ($2$) we repeatedly applied the fact that $\Gamma_{wit}$ is trace-preserving, i.e. $\sum_{k'} \tr \mathcal{R}_{kk'}[\rho_{\rm S}] = \tr \rho_{\rm S}$ for each $k$. Taking the limit of infinite battery $N \rightarrow \infty$ implies $n\rightarrow \infty$, hence the energy change associated with map ${\Gamma}^f_{osc}$ can be expressed as:
\begin{align}
\label{eq:dEW}
\Delta E_W = \langle w \rangle_{osc} = \delta \cdot \left(\sum_{i=0}^{\infty} \tr \mathcal{R}_{01}^i \circ \mathcal{R}_{11}[\rho_{\rm S}] - 1\right),
\end{align}
In Fig. \ref{fig4} we presented a graphical scheme representing the action of map $\Gamma^f_{osc}$ on $\dyad{\epsilon_k}$. Notice that this energy change (and hence work) does not depend on the choice of the initial state $\dyad{\epsilon_k}$. Due to the linearity of $\Gamma_{osc}$ this means that the map would yield the same amount of work if we instead chose any combination of energy levels, i.e. any initial state $\rho_W$. 

The expression for average work (\ref{eq:dEW}) can be further simplified if $\rho_{\rm S}$ is a diagonal state. In that case the action of subchannels $\{R_{kk'}\}$ can be fully described by a set of substochastic matrices $\{R_{kk'}\}$ acting on vectors composed of the diagonal parts of the states. In that case the infinite matrix geometric series in (\ref{eq:dEW}) can be explicitly computed, i.e. $\sum_{i=0}^{\infty} R_{01}^i = (\mathbb{1} - R_{01})^{-1}$. Denoting the diagonal of $\rho_{\rm S}$ with  a vector $\mathbf{x}$ the average work can be expressed as:
\begin{align}
    \langle w \rangle_{osc} = \delta \cdot \left(\mathbf{1}^{\rm T} \left(\mathbb{1} - R_{01}\right)^{-1} R_{11} \mathbf{x} - 1 \right),
\end{align}
where $\mathbf{1}^T = (1, 1, \ldots, 1)$ is the (horizontal) identity vector and $\mathbb{1}$ is the identity matrix. 

\begin{figure*}
 \centering
       \includegraphics[width=0.8\textwidth]{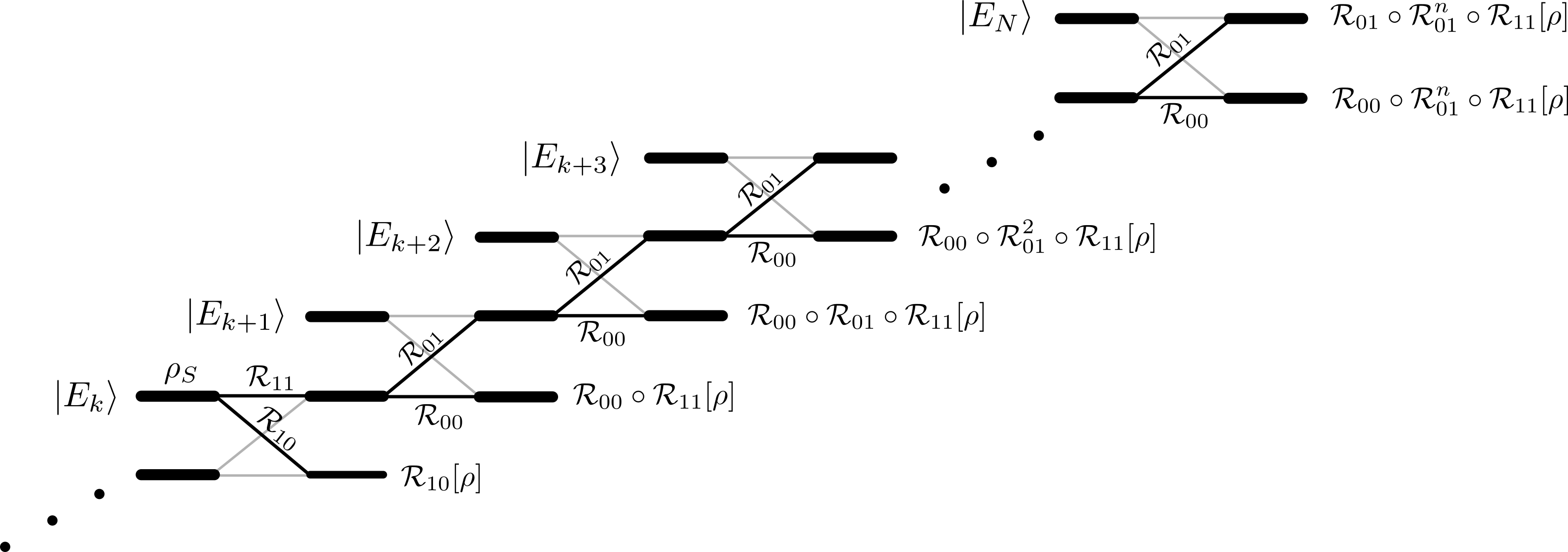}
  \caption{Diagram representing the action of map $\Gamma_N$ on one of the battery eigenstates $\dyad{\epsilon_k}_W$. Black thick lines correspond to the part of channel $\Gamma_N$ which is effectively used in the transformation. Note that the energy spacing between subsequent battery levels is constant and equal to $\Delta$, that is $\epsilon_{k+m} = \epsilon_k + m \cdot \Delta$.}
  \label{fig4}
\end{figure*} 

\subsubsection{The map $\Gamma_{osc}$ satisfies the ETI property}
Here we are going to show that for diagonal states the map $\Gamma_{osc}$ satisfies the ETI property. Hence we are going to show that transition probabilities $\{r(s'k'|sk)\}$ corresponding to $\Gamma_{osc}$ satisfy:
\begin{align}
    \label{eq:eti_2}
    r(s'k'|sk) = r(s',k'+n|s,k+n)
\end{align}
for all $s,s'$ and integers $n$ such that $k'+n \geq 0$ and $k+n  \geq k_{\text{min}}$ with $k_{\text{min} = 1}$. 

Notice first that due to (\ref{mt:rRcorr}) the condition (\ref{eq:eti_2}) can be equivalently written using battery subchannels $\{\mathcal{R}_{kk'}\}$ as:
\begin{align}
    \mathcal{R}_{kk'} = \mathcal{R}_{k+n, k'+n},
\end{align}
for all integers $n$ such that $k'+n \geq 0$ and $k+n  \geq 1$. Consider the battery subchannels associated with the map $\Gamma_{osc}$ for $k \geq k_{\text{min}} = 1$, i.e.:
\begin{align}
    \mathcal{R}_{kk'} &= \sum_{i = 0}^{\infty} \mathcal{R}_{00}\mathcal{R}_{01}^i\mathcal{R}_{11} \delta_{i+k,k'} + \mathcal{R}_{10} \delta_{k-1, k'} = \mathcal{R}_{00}\mathcal{R}_{01}^{k'-k}\mathcal{R}_{11} + \mathcal{R}_{10} \delta_{k-1, k'}
\end{align}
The shifted battery subchannels are given by:
\begin{align}
    \mathcal{R}_{k+n,k'+n} =  \sum_{i = 0}^{\infty} \mathcal{R}_{00}\mathcal{R}_{01}^i\mathcal{R}_{11} \delta_{i+k+n,k'+n} + \mathcal{R}_{10} \delta_{k+n-1, k'+n} = \mathcal{R}_{00}\mathcal{R}_{01}^{k'-k}\mathcal{R}_{11} + \mathcal{R}_{10} \delta_{k-1, k'} = \mathcal{R}_{kk'}.
\end{align}
Hence we can conclude that the map $\Gamma_{osc}$ satisfies the ETI property.

\end{document}